\documentclass{article}
\usepackage{amssymb}
\usepackage{amsthm}
\usepackage{enumerate}
\usepackage{paralist}
\usepackage{xspace}
\usepackage{graphicx}
\usepackage{url}
\usepackage[textsize=footnotesize]{todonotes}
\usepackage{microtype}
\usepackage{stmaryrd}
\usepackage{subfig,caption}
\usepackage[a4paper]{geometry}

\providecommand{\keywords}[1]{{\small \textbf{\textit{Keywords:}} #1}}

\graphicspath{{figures/}}

\newcommand{\rvd}{rectangle visibility representation\xspace}
\newcommand{\rvds}{rectangle visibility representations\xspace}

\newcommand{\RVD}{RVR\xspace}
\newcommand{\RVDs}{RVRs\xspace}

\newcommand{\B}{B}
\newcommand{\W}{W}
\newcommand{\T}{T}
\newcommand{\Bc}{\B-configuration\xspace}
\newcommand{\Wc}{\W-configuration\xspace}
\newcommand{\Tc}{\T-configuration\xspace}
\newcommand{\Bcs}{\B-configurations\xspace}
\newcommand{\Wcs}{\W-configurations\xspace}
\newcommand{\Tcs}{\T-configurations\xspace}
\newcommand{\algo}{\texttt{1PRVDrawer}\xspace}
\newcommand{\foursc}{4-sides-condition\xspace}
\newcommand{\se}{surround-edge\xspace}
\newcommand{\zigzag}{zig-zag-bend-elimination-slide\xspace}
\newcommand{\zz}{zig-zag-slide\xspace}
\newcommand{\rdr}{rectangular dual representation\xspace}

\newif\ifvaremb
\varembtrue
\newtheorem{theorem}{Theorem}
\newtheorem{lemma}{Lemma}
\newtheorem{corollary}{Corollary}
\newtheorem{observation}{Observation}
\newtheorem{definition}{Definition}

\begin{document}

\title{On Visibility Representations of Non-planar Graphs\thanks{Research of Therese Biedl supported by NSERC. Research of Giuseppe Liotta and Fabrizio Montecchiani supported in part by the MIUR project AMANDA ``Algorithmics for MAssive and Networked DAta'', prot. 2012C4E3KT\_001. Research undertaken while Fabrizio Montecchiani was visiting the University of Waterloo supported by NSERC.}}

\author{Therese~Biedl$^1$,
Giuseppe Liotta$^2$,
Fabrizio Montecchiani$^2$
\\[0.1in]
David R. Cheriton School of Computer Science, University of Waterloo, Canada\\
\texttt{\small biedl@uwaterloo.ca}
\\
$^2$Dipartimento di Ingegneria, Universit{\`a} degli Studi di Perugia, Italy\\
\texttt{\small \{giuseppe.liotta,fabrizio.montecchiani\}@unipg.it}
}

\date{}

\maketitle

or
%%%%%%%%%%%%%%%%%%%%%%%%%%%%%%%%%%%%%%%%%%%%%%%%%%%%%%%%%

\maketitle

% ==================================================================
\begin{abstract}
A \rvd (\RVD) of a graph consists of an assignment of axis-aligned rectangles to vertices such that for every edge there exists a horizontal or vertical line of sight between the rectangles assigned to its endpoints.  Testing whether a graph has an \RVD is known to be NP-hard.  In this paper, we study the problem of finding an \RVD under the assumption that an embedding in the plane of the input graph is fixed and we are looking for an \RVD that reflects this embedding. We show that in this case the problem can be solved in polynomial time for general embedded graphs and in linear time for 1-plane graphs (i.e., embedded graphs having at most one crossing per edge). The linear time algorithm uses a precise list of forbidden configurations, which extends the set known for straight-line drawings of 1-plane graphs. These forbidden configurations can be tested for in linear time, and so in linear time we can test whether a 1-plane graph has an \RVD and either compute such a representation or report a negative witness. 
\ifvaremb
Finally, we discuss some extensions of our study to the case when the embedding is not fixed but the \RVD can have at most one crossing per edge.
\fi
\end{abstract}
% ==================================================================

% ==================================================================
\section{Introduction}
% ==================================================================

A visibility representation is an appealing method of displaying graphs.  It consists of assigning axis-aligned rectangles to vertices and horizontal or vertical line segments to edges in such a way that line segments of edges begin and end at their endpoints and intersect no  rectangles in-between.   Equivalently, each edge corresponds to a horizontal or vertical line-of-sight between its endpoints; hence the name ``visibility''.  
Edge segments may cross each other, but any such crossing occurs at a right angle.  These right-angle crossings as well as the absence of bends leads to a high readability of the visualization.    See Fig.~\ref{fig:convert} for an illustration.

Visibility representations were introduced to the Computational Geometry community in 1985, when Wismath \cite{DBLP:conf/compgeom/Wismath85}
showed that every planar graph has a bar-visibility representation, i.e., where vertices are represented as horizontal bars and edges correspond to vertical visibilities.  (The same result was discovered independently multiple times \cite{Duchet1983319,ov-grild-78,DBLP:journals/dcg/RosenstiehlT86,TamassiaTollis86,t-prg-84}.)

Such bar-visibility representations can exist only for planar graphs, so for representing non-planar graphs one must use edge segments in both directions.  (We use the term {\em rectangle visibility representation} or {\em \RVD} to distinguish such drawings from bar-visibility representations.)  Clearly not any graph $G$ can have an \RVD either; for example $G$ must have thickness 2 (i.e., be the union of two planar graphs), and, as was shown in \cite{DBLP:conf/gd/BoseDHS96}, it must have at most $6n-20$ edges.  But neither of these conditions is sufficient.  In fact, Shermer showed in 1996 that it is NP-hard to test whether a given graph has an \RVD~\cite{DBLP:conf/cccg/Shermer96}.

In this paper, we study the problem of testing whether a graph has a rectangle visibility representation if some aspects of it are fixed a priori.  In previous work by Streinu and Whitesides \cite{DBLP:conf/stacs/StreinuW03}, it was shown that testing whether a graph has an \RVD is polynomial if the following additional information is given:  we know for each edge whether it is drawn horizontal or vertical,  we know the rotational scheme (i.e., the fixed order of edges around each vertex), and we know which vertices form the outer face of the visibility representation.  Streinu and Whitesides give necessary and sufficient conditions for when such restrictions can be realized, and show that these can be tested in $O(n^2)$ time.

In this paper, we study a slightly different scenario.  Like Streinu and Whitesides, we assume that the rotational scheme and the outer face is fixed.  However, we do not require to know the directions of the edges; 
%instead we must know a priori which edges cross each other and (if there is more than one crossing on an edge) the order in which these crossings occur along each edge.  
instead we must know a priori which edges cross each other, the order in which these crossings occur along each edge (if there is more than one crossing on an edge). In other words, the graph has a \emph{fixed embedding} (see also Section~\ref{se:preliminaries}).
It turns out that a reduction to the topology-shape-metrics  approach then allows to test in $O(n^{1.5}\log n)$ time whether a graph with $n$ vertices  has a \rvd; we give the details in Section~\ref{sec:shape_metric}.  

As our main result, 
we then consider the case when any edge contains at most one crossing (the so-called {\em 1-planar graphs}).    Previous attempts to create visibility representations for 1-planar graphs had relaxed the restriction that edges must not cross non-incident vertices,  and defined a {\em bar $k$-visibility drawing} to be a bar-visibility drawing where each line of sight can intersect at most $k$ bars
\cite{DBLP:journals/jgaa/DeanEGLST07}.  
In other words, each edge can cross at most $k$ vertices. 
These crossings are arguably less intuitive than crossings between edges. Brandenburg {\em et al.}~\cite{DBLP:journals/jgaa/Brandenburg14} and independently Evans {\em et al.}~\cite{DBLP:journals/jgaa/Evans0LMW14} prove that 1-planar graphs have a bar $1$-visibility drawing (if we are allowed to change the embedding).

In contrast to these result, we consider in Section~\ref{se:characterization} embedding-preserving \rvds of embedded 1-planar graphs, also called 1-plane graphs.   
\ifvaremb
It is easy to see that not all 1-plane graphs have such an \RVD (see also Section~\ref{se:varemb}).  
\fi
The main contribution of our paper is to show that for a 1-plane graph it is possible to test in linear time whether it has an \RVD.  The approach is entirely different from the topology-shape-metrics above and reveals much insight into the structure required for an \RVD to exist.  Specifically, we consider three configurations (called \Bc, \Wc, and the newly defined \Tc) and show that a 1-plane graph has an \RVD if and only if none of these three configurations occurs.  This mirrors in a pleasing way the result by Thomassen~\cite{t-rdg-JGT88}, which shows that a 1-plane graph has a straight-line drawing if and only if it contains no \Bc or \Wc.    Our result also provides an answer (in a special case) to Shermer's question of characterizing when graphs have \RVDs \cite{DBLP:conf/cccg/Shermer96}.
Also, we prove that testing whether a 1-plane graph contains any of the three forbidden configurations can be done in linear time, and in the absence of them we can compute in linear time an \RVD which fits into a grid of quadratic size. It may be worth recalling that embedding-preserving straight-line drawings of 1-plane graphs may require exponential area~\cite{DBLP:conf/cocoon/HongELP12}. 

\ifvaremb
In contrast to planar graphs, 1-planar graphs may have an exponential number of different embeddings
even if 3-connected or 4-connected (see, e.g., Fig.~\ref{fi:expemb}).   Hence our algorithm does not
automatically cover the case when the embedding is not fixed.  We briefly study 1-planar graphs without
fixed embedding in Section~\ref{se:varemb}.  We show that there exist infinitely many 3-connected 
1-planar graphs not admitting a 1-planar embedding realizable as an \RVD, even after removing linearly many 
edges. On the positive side, we show that 4-connected 1-planar graphs, which include the optimal 1-planar 
graphs, always admit an \RVD, except for at most one edge which cannot be represented.
\fi

{\noindent \bf Paper structure.} The remainder of this paper is organized as follows. In Section~\ref{se:preliminaries} we introduce basic definitions and results which will be used throughout the paper. In Section~\ref{se:characterization} we characterize the 1-plane graphs that admit an \RVD. 
\ifvaremb
In Section~\ref{se:varemb} we show both negative and positive results in the variable embedding setting. 
\fi
Finally, in Section~\ref{se:conclusions} we conclude with a discussion and open problems.

% ==================================================================
\section{Definitions and Basic Results}\label{se:preliminaries}
% ==================================================================

\subsection{Embedding-preserving Visibility Representations}

We only consider \emph{simple graphs}, i.e., graphs with neither loops nor multiple edges. A \emph{drawing} $\Gamma$ of a graph $G$ maps each vertex to a point of the plane and each edge to a Jordan arc between its two endpoints. We only consider \emph{simple drawings}, i.e., drawings such that the arcs representing two edges have at most one point in common, which is either a common endpoint or a common interior point where the two arcs properly cross. A drawing is \emph{planar} if no two arcs cross. 
A drawing divides the plane into topologically connected regions, called \emph{faces}. The unbounded region is called the \emph{outer face}. 
%A \emph{planar embedding} of a graph is an equivalence class of planar drawings that define the same set of faces. A graph with a given planar embedding is a \emph{plane graph}.
For a planar drawing the boundary of a face consists of vertices and edges, while for a non-planar drawing the boundary of a face may contain vertices, crossings, and edges (or parts of edges). 
An \emph{inner} face/edge/vertex is a face/edge/vertex that is not (part of) the outer face.

A {\em rotational scheme} of a graph $G$ consists of assigning at each vertex a cyclic order of the edges incident to that vertex.  A drawing is said to respect a rotational scheme if scanning around the vertex in clockwise order encounters the edges in the prescribed order.  We say that a graph has a {\em fixed embedding} if the graph comes with a rotational scheme, a fixed order of crossings along each edge, and with one face indicated to be the outer face.  (One can show that fixing the rotational scheme and the order of crossings fixes all faces.)  
Given a graph $G$ with a fixed embedding, the \emph{planarization} 
$G_p$ of $G$ is the graph obtained by replacing each crossing point of $G$ with 
a \emph{dummy-vertex}. The vertices of $G$ in $G_p$ are called \emph{original vertices} to avoid confusion.   The planarization $G_p$ is a {\em plane} graph, i.e., a planar graph with a fixed embedding without crossing.

A {\em visibility representation} (also called {\em \rvd} or \RVD for short) of a graph $G$ is an assignment of disjoint axis-aligned rectangles to the vertices in such a way that for every edge there exists a horizontal or vertical line of sight between its endpoints.
We follow a commonly adopted model where the lines of sight are {\em thick}, i.e., have non-zero area (see e.g.~\cite{DBLP:journals/ipl/KantLTT97,DBLP:conf/stacs/StreinuW03,TamassiaTollis86,DBLP:conf/compgeom/Wismath85}).  Thus we may assume that each edge is routed as a (horizontal or vertical) line segment, attaching at each end at a point that is not a corner of the rectangle of that vertex.  
In what follows, when this leads to no confusion, we shall use the same term \emph{edge} to indicate both an edge of a graph, the line of sight between two rectangles, and the line segment representing the edge, and we use the same term \emph{vertex} for both the vertex of a graph and the rectangle that represents it. We observe that if we have an \RVD, then we can naturally extract a drawing from it as follows.  Place a point for each vertex $v$ inside the rectangle representing $v$ and connect it to all the attachment points of incident edges of $v$ inside the rectangle; this adds no crossing. 
%\footnote{The resulting drawing is what has been called a {\em RAC-drawing}: Vertices are points, edges are curves, and crossings occur only at right angles.  Our constructive result (Lemma~\ref{le:algo}, Theorem~\ref{th:ve-4conn} and Corollary~\ref{co:ve-optimal}) 
%hence also gives RAC-drawings with at most 2 bends per edge.}
 We say that an \RVD $\Gamma$ of a graph $G$ \emph{respects} the embedding of $G$ if applying this operation to $\Gamma$ results in the same embedding. 

\subsection{Topology-shape-metrics}

The {\em topology-shape-metrics} approach is a well-known method introduced
by Tamassia \cite{Tam87} to create
orthogonal drawings of graphs, i.e., drawings where all vertices are points
and edges are curves with horizontal and vertical segments.    In this
section, we briefly review how any \RVD can be expressed as an
orthogonal-shape-metric; this is quite straightforward but has to our
knowledge not been used before.  We will use such metrics twice in
our paper: once to determine whether a given embedded graph has an \RVD,
and once to store an \RVD obtained for 4-connected 1-plane graphs so that
they can be manipulated efficiently.

Consider an \RVD $\mathcal R$ and interpret it as a PSLG, i.e., a planar straight-line graph in the computational geometry sense.   Thus, replace every corner of a rectangle, every attachment point of an edge at a rectangle, and every crossing by a node of the PSLG, and connect nodes with arcs if and only if they are connected by (part of) an edge segment or (part of) the boundary of a rectangle.  
Each node of the PSLG has around it a sequence of angles that are all multiple of $90^\circ$ and sum to $360^\circ$.  Likewise each face of the PSLG has around it a sequence of angles that are all multiple of $90^\circ$ and that sum to $(k-2)180^\circ$ for inner faces and $(k+2)180^\circ$ for the outer face, where $k$ is the number of nodes on the face.    
This motivates the following definition:

\begin{figure}[t]
\centering
\hspace*{\fill}
\includegraphics[width=0.3\linewidth,page=1]{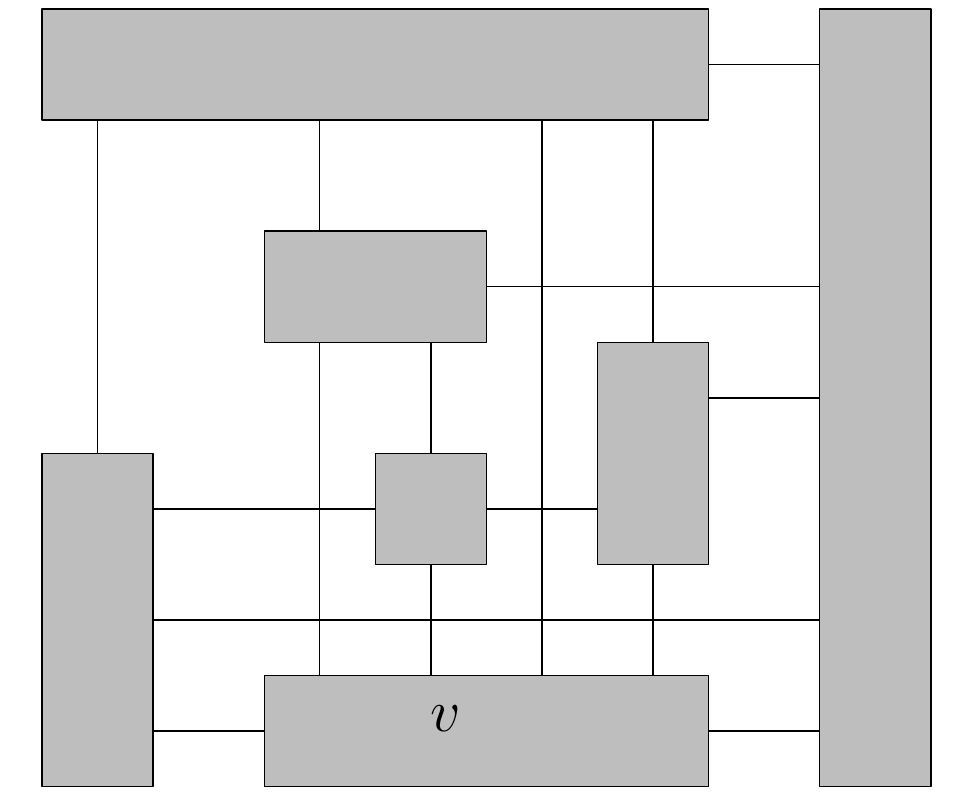}
\hspace*{\fill}
\includegraphics[width=0.3\linewidth,page=2]{convert.pdf}
\hspace*{\fill}
\includegraphics[width=0.3\linewidth,page=3]{convert.pdf}
\hspace*{\fill}
\caption{A visibility representation of a graph, the orthogonal representation obtained from it, and the orthogonal representation used to compute it.  We show only selected non-zero bend-counts.}
\label{fig:convert}
\label{fig:VRtoMetric}
\end{figure}

\begin{definition}
An {\em orthogonal representation} is a plane graph $H$ with maximum degree 4, an assignment of an angle $\alpha(v,f)\in \{90^\circ,180^\circ,270^\circ\}$ to any vertex-face-incidence (between vertex $v$ and face $f$) of $H$, and a bend-count $b(e,f)\in \{0,+1,-1,+2,-2,\dots\}$ to any edge-face-incidence (between edge $e$ and face $f$) of $H$, such that the following is satisfied:
\begin{itemize}
\item at every vertex $v$, the angles at $v$ sum to $360^\circ$,
\item at every edge $e$, the two bend-counts sum to 0,
\item for every inner face $f$ with $k$ vertices, the angles and bend-counts at $f$ sum to $(k-2)180^\circ$.  Formally,
	$$ \sum_{(v,f)} \alpha(v,f) + \sum_{(e,f)} b(e,f)\cdot 90^\circ = (k-2)180^\circ$$
\item the angles and bend-counts at the outer face sum to $(k+2)180^\circ$, where $k$ is the number of vertices on the outer face.
\end{itemize}
\end{definition}

We have shown above:

\begin{observation}
\label{obs:convert}
If $G$ has an \RVD $\Gamma$, then construct a plane graph $H$ by replacing every crossing, every edge-vertex attachment point and every corner of a vertex-rectangle of $\Gamma$ with a node in $H$.  The resulting graph then has an orthogonal representation where all bend-counts are 0.
\end{observation}

Tamassia showed \cite{Tam87} that for any orthogonal representation, one can reconstruct a planar orthogonal drawing of the graph that respects these angles and bend-counts.  In other words, at any vertex-face-incidence the angle is $\alpha(v,f)$, and any edge $e$ has $|b(e,f)|$ bends, where the bends have angle $90^\circ$ in the face where the bend-count at $e$ is positive and hence angle $-90^\circ=270^\circ$ in the other face. In particular, it is not required to know the edge-lengths a priori.  This reconstruction process takes linear time in the size of $H$ and the total number of bend-counts.  If the orthogonal representation came from an \RVD, then the resulting drawing has the exact same set of crossings and the face that came from a vertex rectangle is again rectangular, so the resulting drawing is again an \RVD. 

Tamassia also showed how to compute for a given plane graph $H$ some angle-assignment and bend-counts such that the result is an orthogonal representation.  This is done via a flow-approach; see \cite{Tam87} or \cite{dett-gdavg-99} for details.  If one uses a minimum-cost-flow algorithm, then this allows to compute the topological orthogonal shape that minimizes the total number of bends.  Of higher interest to us is instead the following variant:  One can specify upper and lower bounds on the angles and bend-counts, and can then use a maximum-flow algorithm to test whether a feasible flow exists, and hence to test whether there exists an orthogonal representation for which angles and bend-counts are within the specified bound.  
The first version can be tested in $O(n^{1.5})$ time while the second version in $O(n^{1.5}\log{n})$ time, where $n$ is the number of vertices of $H$
\cite{DBLP:journals/jgaa/CornelsenK12}.

\subsection{Testing embedded graphs for \RVDs}
\label{sec:shape_metric}

Now presume that we are given a graph $G$  with a fixed embedding.
In this section, we will show how to test whether $G$ has an \RVD that respects its embedding using the topology-shape-metrics approach.  Note that 
%the rotational scheme and the order of crossings 
the embedding tells us nearly the entire graph $H$ of the topology of a (putative) \RVD of $G$.  The only thing unknown is the location of the corners of each rectangle of a vertex.  We can now interpret these corners as bends (rather than as nodes) in the topology.   Hence define a graph $H$ as follows (see also Fig.~\ref{fig:VRtoMetric}):
\begin{itemize}
\item Start with the planarization $G_p$ of $G$.
\item For any original vertex $v$ of $G$, define a cycle $C_v$ with $\deg(v)$ nodes in $H$.  Attach the (parts of) edges incident to $v$ to the nodes of cycle $C_v$ in the specified order of the planar embedding, in such a way that the interior of the cycle forms a face.  In other words, form a cycle that can represent the boundary of the rectangle of $v$.
\item Impose an upper and lower bound of 0 onto the bend-counts of all (parts of) edges of $G$. In other words, we do not allow any bends for the original edges.
\item Impose a lower bound of 0 onto each bend-count $b(e,f_v)$ where $e$ is an edge of some cycle $C_v$ and $f_v$ is the face formed by cycle $C_v$.  In other words, cycle $C_v$ may have bends, but all such bends must form $90^\circ$ angles towards $f_v$.
%\item Impose an upper and lower bound of $90^\circ$ on the angle of any vertex-face-incidence $(v,f)$ where $v$ is a dummy-vertex.  In other words, crossings must remain crossings.  In fact, these bounds are not even needed, because the angle-sum-condition around $v$ will automatically create $90^\circ$ angles at all these incidences.
% TB: I made the 180 below a lower AND upper bound.  Not actually needed, but doesn't hurt.
\item If $a$ is a node of $C_v$, then impose an upper and lower bound of $180^\circ$ onto the angle $\alpha(a,f_v)$, where $f_v$ as before is the face inside $C_v$.  In other words, at attachment points of edges face $f_v$ must not have a corner.
\end{itemize}

\begin{theorem}
%Let $G$ be any graph with a fixed order of edges around each vertex, and a fixed order of crossings along each edge.  Then we can test in $O((n+c)^{5/2})$ time whether $G$ has an \RVD that respects these orders of edges and crossings.  The time drops to $O((n+c)^{3/2})$ if additionally the outer-face is fixed.
Let $G$ be  an $n$-vertex graph with a fixed embedding. Then we can test in $O(n^{1.5}\log {n})$ time whether $G$ has an \RVD that respects this embedding. 
\end{theorem}
\begin{proof}
Build the graph $H$ as described above.  It should be obvious that if we can find an orthogonal representation of $H$ that respects these constraints, then it gives rise to an \RVD of $G$.  Indeed, for any vertex $v$ the cycle $C_v$ must form a rectangle:  It cannot have a bend at any attachment-point, and any bend that it has provides an angle of $90^\circ$, hence by the condition it has exactly four such bends.  At any dummy-vertex the four angles must be $90^\circ$ by the condition for nodes.  All edges of $H$ that came from edges of $G$ are drawn straight, and continue straight at any crossing, so give rise to a line of sight. We then use Tamassia's flow-approach to test whether an orthogonal representation  of $H$ exists that satisfies all constraints. This takes $O(n^{1.5}\log{n})$ time~\cite{DBLP:journals/jgaa/CornelsenK12} ($G$ has $O(n)$ crossings).  
\end{proof}

% ==================================================================
\section{Embedded 1-planar graphs}\label{se:characterization}
% ==================================================================

While the topology-shape-metrics approach gives a polynomial way to test whether a graph has an  
\RVD, its black-box approach is somewhat unsatisfactory.  In particular, in case of a
negative answer, we obtain no good insights as to which parts of the graph prevented
the existence of a representation.    In this section, we therefore turn
our attention to 1-planar graphs.  We chose this graph class for two reasons.  One is
that they are known to have thickness 2 and at most $4n-8$ edges~\cite{pt-gdfce-C97}, and so they are 
good candidates for always having an \RVD.  (We will, however, see that this is not the
case.) 
Secondly,  1-planar graphs have been studied widely in the graph theory and graph drawing communities (see e.g.~\cite{DBLP:conf/gd/AlamBK13,DBLP:journals/jgaa/Brandenburg14,DBLP:journals/combinatorics/CzapH13,DBLP:journals/jgaa/Evans0LMW14,DBLP:conf/cocoon/HongELP12,DBLP:journals/dm/Korzhik08a,pt-gdfce-C97,Suzuki2010,t-rdg-JGT88}), and visibility representations of 1-planar graphs hence should be of interest. 

% ==================================================================
\subsection{Background}
% ==================================================================

A graph is {\em 1-planar} if it has a drawing with at most one crossing per edge.
A {\em 1-plane} graph is a 1-planar graph with a fixed embedding.  Recall
that this means that we have a fixed order of edges at each vertex, it is
fixed which pairs of edges cross, and we know what the faces are and which
one of them is the outer face. A 1-plane graph $G$ is called 
{\em triangulated 1-plane} if all faces are triangles.  In a 1-plane graph,
a facial triangle is composed of 
three vertices or two vertices and a crossing point. 

\begin{figure}[tb]
    \centering
    \begin{minipage}[b]{.2\textwidth}
    	\centering
    	\subfloat[\label{fi:bgraph}{\B}]
	    {\includegraphics[scale=0.8,page=2]{configurations}}
    \end{minipage}
    \hfill
    \begin{minipage}[b]{.2\textwidth}
    	\centering
    	\subfloat[\label{fi:xgraph}{Kite}]
	    {\includegraphics[scale=0.8,page=1]{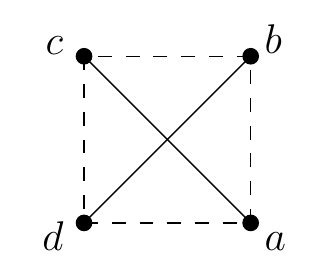}}
    \end{minipage}
    \hfill
    \begin{minipage}[b]{.2\textwidth}
    	\centering
    	\subfloat[\label{fi:wgraph}{\W}]
	    {\includegraphics[scale=0.8,page=3]{configurations}}
    \end{minipage}
    \hfill
    \begin{minipage}[b]{.2\textwidth}
    	\centering
    	\subfloat[\label{fi:tgraph}{\T}]
	    {\includegraphics[scale=0.8,page=4]{configurations}}
    \end{minipage}
    \caption{Possible crossing configurations in a 1-plane graph $G$. (a) A \Bc.  (b) A (subgraph) of a kite. The dashed edges either do not exist or define faces. (c) A \Wc. (d) A \Tc.}
\end{figure}

Let $e=(a,c)$ and $e'=(b,d)$ be two edges of $G$ that cross at a point $p$.  We say that $e,e'$ induce a {\em \Bc} if there exists an edge between their endpoints (say edge $(a,b)$) such that the triangle $\{a,b,p\}$ contains vertices $c$ and $d$ inside (see also Fig.~\ref{fi:bgraph}).

If the crossing is not in a \Bc, then there are two further possibilities for edge $(a,b)$.  The triangle $(a,b),(b,p),(p,a)$ may be an interior face, or it could have other vertices inside.
% TB: removed the reference to 3-connected, since we haven't talked about that at all yet
, even in a 3-connected graph since edge $(a,b)$ may be crossed.  
If, for all pairs in $\{a,c\}\times \{b,e\}$, there either is no edge between the vertices or the triangle that it forms is a face,
then we could add the edges that did not exist previously, routing them along the crossing, and then obtain a {\em kite}, i.e., a crossing with a 4-cycle among its endpoints such that removing the crossing turns the 4-cycle into a face.

Let $(a,g)$ and $(b,d)$ be two edges of $G$ that cross at a point $p$, and let $(a,f)$ and $(b,e)$ be two further edges of $G$ that cross at a point $q$. The four edges induce a \emph{\Wc} if vertices $c,d,e,f$ lie inside the closed region delimited by the edge segments $(a,p)$, $(b,p)$, $(a,q)$, and $(b,q)$ (see also Fig.~\ref{fi:wgraph}). 

For the result in this paper, we need to
introduce a forth configuration, called the \emph{trillium configuration}, or \emph{\Tc} for short, which is illustrated in Fig.~\ref{fi:tgraph}. Namely, let $(a,c)$ and $(b,d)$ be a pair of edges of $G$ that cross at a point $p$. Let $(a,e)$ and $(c,h)$ be a second pair of edges that cross at a point $q$. Let $(b,f)$ and $(c,i)$ be a third pair of edges that cross at a point $t$. The six edges induce a \Tc if vertices $d$, $e$, $f$, $g$, $h$, $i$ lie inside the closed region delimited by the edge segments $(a,p)$, $(b,p)$, $(a,q)$, $(c,q)$, $(b,t)$, and $(c,t)$. Vertices $a,b,c$ are called the \emph{outer vertices} of the \Tc, while the remaining six vertices are the \emph{inner vertices}.

% ==================================================================
\subsection{Main result and outline}
% ==================================================================

In this section we give a characterization of 1-plane graphs that admit an \RVD. Our characterization is summarized by the following theorem.

\begin{theorem}\label{th:characterization}
A 1-plane graph $G$ admits an \RVD if and only if it contains no \Bc, no \Wc, and no \Tc.
\end{theorem}

Notably, Theorem~\ref{th:characterization} extends the characterization for straight-line drawability of 1-plane graphs given by Thomassen~\cite{t-prg-84}, adding the \Tc to the set of obstructions. The necessity of the condition is proved via an easy angle-counting-argument
in Section~\ref{sse:necessary}.
The sufficiency-proof is quite intricate and is given in Section~\ref{sse:sufficient}.  It comes with
an efficient algorithm to construct an \RVD $\Gamma$ of a drawable graph $G$. 
Finally, we describe a testing procedure (Section~\ref{sse:testing}) that easily follows from our characterization. Next theorem is obtained by combining our testing procedure and our drawing algorithm. 

\begin{theorem}\label{th:test}
Let $G$ be a 1-plane graph with $n$ vertices. There exists an $O(n)$-time algorithm to test whether $G$ admits an \RVD. Also, in the positive case the algorithm computes an \RVD of $G$ that has $O(n^2$) area.
\end{theorem}

\subsection{Proof of necessary condition}\label{sse:necessary}

Our first goal is to show that no \B-/\W-/\Tc can be a subgraph of a 1-plane graph admitting an \RVD. To this aim we show a slightly stronger statement.

\begin{lemma}\label{le:outerface}
Let $G$ be a graph that admits an \RVD $\Gamma$ and whose outer face is a cycle $\mathcal C$.  If while walking along $\mathcal C$ we encounter $k$ vertices and $c$ crossing points, then $k\geq 3$ and $c\leq 2k-4$.
\end{lemma}
\begin{proof}
Convert $\Gamma$ into an orthogonal representation (Observation~\ref{obs:convert}) and consider the angles at the outer face.
Since there are $k$ vertices on the outer face, we have $2k$ edge-attachment points, and each contributes $90^\circ$.
Each of the $c$ crossings contributes $90^\circ$.  If $b$ is the number of rectangle-corners on the outer face, then $b\leq 4k$,
and each contributes $270^\circ$.   In total the outer face has $2k+c+b$ nodes and the sum of angles must be $180^\circ(2k+c+b+2)$.
But we also know that the angles sum to $90^\circ(2k+c)+270^\circ b = 180^\circ(2k+c+b+2)+90^\circ(b-2k-c-4)$.  So
$b-2k-c-4=0$, hence $4k\geq b=2k-c-4$, hence $c\leq 2k-4$.  This implies $k\geq 2$, and in fact $k\geq 3$ is required, else $c=2$ and the outer face would consist of a double edge.
%\iffalse
%\begin{figure}[t]
%\centering
%\includegraphics[scale=0.5]{boundary}
%\caption{Illustration for the proof of Lemma~\ref{le:outerface}.}\label{fi:outerface}
%\end{figure}
%
%Recall that all lines of sight are thick.  We can hence represent each edge by a line in such a way that no such line ends exactly at the corner of a vertex rectangle.  Consider the polygon $P$ obtained while walking along the outer face, hence along $\mathcal C$, in clockwise direction.
%The corners of $P$ can be $(i)$ corners of rectangles (there are at most $4k$ of these) , $(ii)$ crossing points (there are $c$ of these), $(iii)$ attachment points between edges and rectangles (there are $2k$ of these). At each crossing point and each attachment point the boundary makes a leftward $90^\circ$ turn.  At each corner of a rectangle it makes a rightward $90^\circ$ turn. See also Fig.~\ref{fi:outerface} for an illustration. 
%
%In total the boundary of $P$ needs to make a turn of $360^\circ$ degrees rightward.   By the above, the total angle turned rightward is at most $(4k-c-2k)90^\circ = (2k-c)90^\circ$, and this must be at least $360^\circ$.  Hence $2k-c\geq 4$, which implies $c\leq 2k-4$ as desired.  Also, $\mathcal C$ must have at least 3 vertices or crossings, so $3\leq k+c \leq 3k-4$, hence $k\geq 3$.
%\fi
\end{proof}

Since a \Bc (\Wc, \Tc) has as outer face a cycle with $k=2$ and $c=1$ ($k=2$ and $c=2$, $k=3$ and $c=3$), the following corollary follows.

\begin{corollary}\label{co:nobwt}
Let $G$ be a 1-plane graph that admits an \RVD. Then $G$ contains no \Bc, no \Wc, and no \Tc. 
\end{corollary}

\subsection{Proof of sufficient condition}\label{sse:sufficient}

We show that the absence of any \B-/\W-/\Tc suffices to compute an \RVD. Let $G$ be a 1-plane graph with $n$ vertices and with no \B-/\W-/\Tc as a subgraph, see also Fig.~\ref{fi:g} for an example. We describe a drawing algorithm, \algo, which computes an \RVD $\Gamma$ of $G$. We first assume that $G$ is 3-connected (Sections~\ref{ssse:triangulate}--\ref{ssse:undoing}), and then extend the proof to the general case (Section~\ref{ssse:general}).

%The high-level idea is as follows:  Triangulate and planarize $G$ to obtain $G^+$, decompose $G^+$ into its 4-connected components, compute an \RVD for each such component $C$ by modifying a rectangular dual representation, and then suitably merge the drawings.   All steps maintain the planar embedding of $G^+$, which was chosen such that it contains the 1-plane embedding of $G$.  Due to special conditions of the drawing, we are then able to undo the planarization and replace dummy vertices by crossings without requiring bends.  The final \RVD of $G$ is then obtained by deleting all edges that were added.

\subsubsection{Triangulating and planarizing}
\label{ssse:triangulate}

The first step of algorithm \algo is to triangulate $G$.   This is a well-known operation for 3-connected 1-plane graphs, see \cite{DBLP:conf/gd/AlamBK13,DBLP:conf/cocoon/HongELP12}.  However, these algorithms all modify the given 1-planar embedding, even if there is no \Bc or \Wc. In fact, this may be required since every crossing must form a kite in a triangulation of $G$, and so if an edge between endpoints of a crossing already existed, but did not form a face, then it must be re-routed.  We do not want to change the embedding, and hence cannot triangulated the 1-plane graph per se.  Instead, we combine the triangulation step with a planarization step to ensure that the result has no multiple edges, yet at all crossings (respectively the corresponding dummy-vertices) we have kites. 

We construct the triangulated planar graph $G^+$ with the following steps:
\begin{itemize}
\item Let $G_1$ be the planarization of $G$.  Note that dummy-vertices have degree 4 (and we will maintain this throughout later modifications).
\item  For each dummy-vertex $z$ of $G_1$, let $v_0,\dots,v_3$ be the four neighbors in clockwise order.  For $i=0,\dots,3$, consider the face at $z$ between the edges $(z,v_i)$ and $(z,v_{i+1})$ (addition modulo 4).  If this face is the outer face, or an inner face that is not a triangle, then add the edge $(v_i,v_{i+1})$.  Put differently, we add the edges that are needed to turn $z$ into a kite.  (In what follows, we will use the previously defined terms ``kite'', ``\Bc'', ``\Wc'' and ``\Tc'' for the planarized version as well, i.e., for the same situations with crossings replaced by dummy-vertices.)
Call the resulting graph $G_2$.  The following is crucial for our correctness:

% TB: specifically added the "3-connected" here, since it is crucial
\begin{lemma}
If $G$ is simple and 3-connected and has no \Bc, \Wc or \Tc, 
then $G_2$ has no \Bc, \Wc, or \Tc and no multiple edges.
\end{lemma}
\begin{proof}
Every dummy-vertex in $G_2$ corresponds to a crossing in $G$.  Hence for any \Wc or \Tc in $G_2$, the corresponding crossings in $G$ would define a \Wc or \Tc in $G$ since we did not change the embedding or outer face.  Now assume that $G_2$ has a \Bc, say at dummy-vertex $z$ with neighbors $a,b,c,d$ and with edge $(a,b)$ such that triangle $\{a,b,z\}$ contains $c$ and $d$ inside and hence is not a face.
%If $(a,b)$ also existed in $G$, then this was a \Bc at the corresponding crossing.  If $(a,b)$ did not exist in $G$, then it was added to $G_2$ due to some other dummy-vertex $y$ that also had $a,b$ as neighbors, in such a way that $\{a,b,y\}$ becomes a face.  
We also surrounded $z$ by a kite, which adds an edge $(a,b)$ for which $\{a,b,z\}$ is a face.  Hence there are two copies of edge $(a,b)$.  

It remains to show that there cannot be any double edge (whether created from a 
\Bc or otherwise).  Assume for contradiction that there is some double edge $(u,v)$.  It is not possible that the two copies of $(u,v)$ form an inner face, because $G$ was simple, and before adding an edge we always check whether such an edge exists already at this place in the planar embedding.  So the two copies either form the outer face or form no face at all.  In the latter case, $\{u,v\}$ is a cutting pair of the planar graph $G_2$, and (as one easily shows) therefore also of $G$.  This contradicts 3-connectivity.%
%\footnote{This is the only place where 3-connectivity is needed.  Hence our construction works for any graph that has a simple planarized super-graph without forbidden configuration for which dummy-vertices are surrounded by kites.}

So multiple edges can only exist if the two copies form the outer face.  At least one of the copies did not exist in $G$, and hence was added due to some dummy-vertex $z$.  If the other copy already existed in $G$, then the crossing at $z$ together with this copy forms a \Bc in $G$.  If the other copy did not exist in $G$, then it was added due to some dummy-vertex $y$, and the crossings at $z$ and $y$ form a \Wc.  
\end{proof}

\item 
In the plane graph $G_2$ each dummy-vertex is surrounded by a kite by construction. Hence all faces of $G_2$ that are not triangles contain only original vertices.  Triangulate each such face arbitrarily; again this does not create a \Bc or multiple edges since in a planar graph this would imply a cutting pair.  
\end{itemize}
The resulting graph $G^+$ is the plane triangulated graph that we use to obtain the \RVD. By construction it has the following crucial properties:  It has no \Bc, \Wc or \Tc, every dummy-vertex has degree 4, and removing from $G^+$ all the edges inserted in the above procedure, and replacing each dummy-vertex with a crossing point, we obtain the original 1-plane graph $G$.   

\begin{figure}[t]
    \centering 
    \begin{minipage}[b]{.25\textwidth}
    	\centering
    	\subfloat[\label{fi:g}{$G$}]
	    {\includegraphics[scale=0.35,page=1]{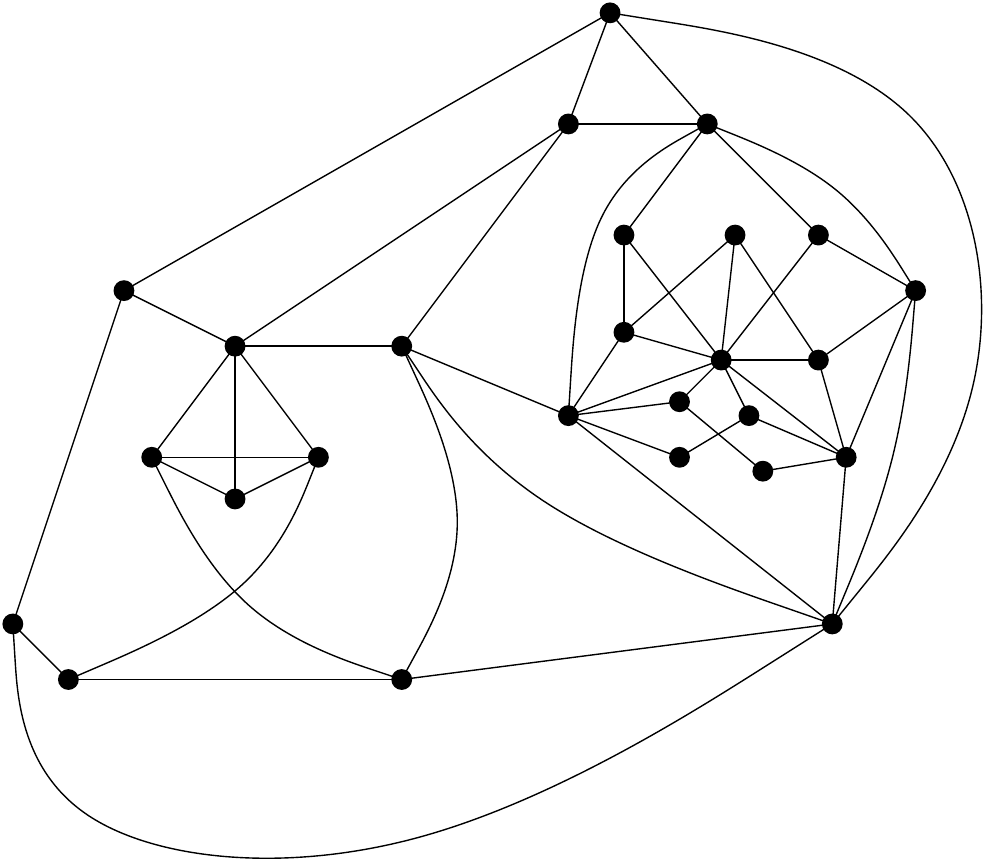}}
    \end{minipage} 
    \hfil
    \begin{minipage}[b]{.25\textwidth}
    	\centering
    	\subfloat[\label{fi:gplus}{$G^+$}]
	    {\includegraphics[scale=0.35,page=2]{3conn1plane}}
    \end{minipage}
    \hfil
    \begin{minipage}[b]{.4\textwidth}
    	\centering
    	\subfloat[\label{fi:4blocktree}{$\mathcal T$}]
	    {\includegraphics[scale=0.4,page=3]{3conn1plane}}
    \end{minipage}
    \caption{(a) A 1-plane graph $G$. (b) The triangulated planar graph $G^+$ obtained from $G$. Dummy-vertices are pink squares. (c) The 4-block tree $\mathcal T$ of $G^+$.}  
\end{figure}

\subsubsection{The 4-block tree and an outline}\label{ssse:4blocktree}

Algorithm \algo uses a decomposition of $G^+$ into its 4-connected components. 
Here, a 4-connected component $C$ of $G^+$ is a 4-connected triangulated subgraph of $G^+$. Let the \emph{4-block tree} $\mathcal T$ of $G^+$ be a tree defined as follows (see also~\cite{DBLP:journals/ijcga/Kant97} and Fig.~\ref{fi:gplus}).
 Every 4-connected component $C_\nu$ of $G^+$ is represented by a node $\nu$ in $\mathcal T$. There is an edge between two nodes $\nu$ and $\mu$ in $\mathcal T$, if there is a {\em separating triangle} (i.e., a triangle with vertices both inside and outside) that belongs to both $C_\nu$ and $C_\mu$. 
We root $\mathcal T$ at the node $\rho$ with the 4-connected component that contains the outer face.  Then for any parent $\nu$ and child $\mu$, 
the separating triangle common to $C_\nu$ and $C_\mu$ is an inner face in $C_\nu$ and the outer face of $C_\mu$. 
%Fig.~\ref{fi:4blocktree} shows the 4-block tree of the graph in Fig.~\ref{fi:gplus}.  
It is known that $\mathcal T$ can be computed in $O(n)$ time \cite{DBLP:journals/ijcga/Kant97}.

We now give an overview of the remaining steps of algorithm \algo. It first visits $\mathcal T$ top-down and determines for each 4-connected component $C_\nu$ a special edge, called {\em \se} of $C_\nu$. It also computes an \RVD $\gamma_\nu$ of $C_\nu$.   Next, the algorithm visits $\mathcal T$ bottom-up.  Let $\nu$ be a node of $\mathcal T$ with $h \geq 1$ children  $\mu_1, \dots, \mu_h$. Denote by $G_\nu$ the graph whose 4-block tree is the subtree of $\mathcal T$ rooted at $\nu$. The already computed  visibility representations $\Gamma_{\mu_1}, \dots, \Gamma_{\mu_h}$ of $G_{\mu_1}, \dots, G_{\mu_h}$ are suitably merged into $\gamma_\nu$ and an \RVD $\Gamma_\nu$ of $G_\nu$ is obtained. 

At the root we obtain an \RVD of $G_\rho=G^+$.  In order to turn this in an \RVD of $G$, we must un-planarize, i.e., replace every dummy-vertex by a crossing.  For this to be feasible, we need each dummy-vertex is drawn in a special way:  its rectangle needs to have exactly one incident edge on each side.  We call this the {\em \foursc}. Even then converting dummy-vertices into crossings is non-trivial, since the four edges on the four sides may not be suitably aligned.  To achieve such an alignment, we use the so-called ``\zigzag'' \cite{BLPS13}, with which we can transform parts of the drawing until edges are aligned and hence each dummy-vertex can be turned into a crossing point.  Since this introduces no other crossings we obtain an \RVD, and deleting the added edges gives the desired \RVD of $G$.

\subsubsection{Transversal pairs of bipolar orientations}\label{ssse:transversal}

We will need a few definitions. 
An \emph{internally 4-connected plane graph} (also known as \emph{irreducible triangulation}~\cite{DBLP:journals/dm/Fusy09}) is a plane graph where all inner faces are triangles and that has no separating triangle. Let $G$ be an internally 4-connected planar graph for which the outer face is a 4-cycle.
A \emph{transversal pair of bipolar orientations} of $G$ (also known as {\em regular edge labeling}~\cite{DBLP:journals/tcs/KantH97}, or simply \emph{transversal structure}~\cite{f-rsrpg-13}) assigns to each inner edge one of two colors\footnote{In all figures of the paper red edges appear in lighter gray than blue edges when printed b/w.}, say \emph{red} and \emph{blue}, and an orientation such that: $(i)$ For each inner vertex $v$, the clockwise order of edges around $v$ contains outgoing blue edges, then outgoing red edges, then incoming blue edges, then incoming red edges, and none of these sets is empty.
$(ii)$ Let $v_N$, $v_E$, $v_S$, $v_W$ be the four vertices on the outer face of $G$, in clockwise order. All inner edges incident to $v_N$/$v_E$/$v_S$/$v_W$ are blue incoming/red incoming/blue outgoing/red outgoing.  The edges on the outer face can be colored arbitrarily and are oriented so that both the red and the blue graph are acyclic.
See e.g. Fig.~\ref{fi:transversal}, where $v_N=v$, $v_E=x$, $v_S=w$, $v_W=u$.

Such a transversal pair of bipolar orientations
always exists and can be computed in linear time \cite{DBLP:journals/dm/Fusy09,DBLP:journals/tcs/KantH97}.
%, presuming as before that the graph is
%internally 4-connected with a 4-cycle as outer face.   
The same papers also show how to convert, in linear time,
a transversal pair of bipolar orientations into 
a {\em \rdr}.  This is a representation of the graph that assigns
interior-disjoint rectangles to vertices such that the union of these rectangles
is a rectangle without holes.  For each edge $(v,w)$, the rectangles of $v$ and $w$ share a positive-length part of their boundaries.  Moreover, if the edge is directed $v\rightarrow w$, then the rectangle of $v$ is left (below) $w$ if and only if the edge is red (blue).
See e.g. Fig.~\ref{fi:rd}.

\subsubsection{Surround-edges}
%\label{ssse:surrounded}

\begin{figure}[t]
    \centering 
    \begin{minipage}[b]{.24\textwidth}
    	\centering
    	\subfloat[\label{fi:gnuminus}{$C^-_\nu$}]
	    {\includegraphics[scale=0.55,page=1]{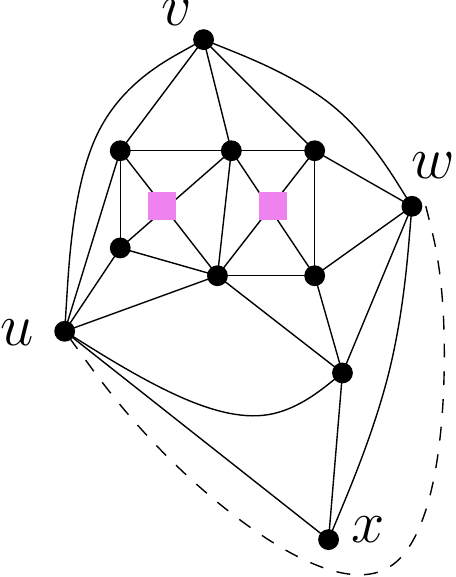}}
    \end{minipage} 
    \hfill
    \begin{minipage}[b]{.24\textwidth}
    	\centering
    	\subfloat[\label{fi:transversal}{transv.pair}]
	    {\includegraphics[scale=0.55,page=2]{construction_RD}}
    \end{minipage}
    \hfill
    \begin{minipage}[b]{.24\textwidth}
    	\centering
    	\subfloat[\label{fi:rd}{$\mathcal R$}]
	    {\includegraphics[scale=0.4,page=3]{construction_RD}}
    \end{minipage}
    \hfil
     \begin{minipage}[b]{.24\textwidth}
    	\centering
    	\subfloat[\label{fi:rnu}{$\mathcal R_\nu$}]
	    {\includegraphics[scale=0.4,page=5]{construction_RD}}
    \end{minipage} 
%    \hfil
%     \begin{minipage}[b]{.19\textwidth}
%    	\centering
%    	\subfloat[\label{fi:rnu2}{The \RVD.}]
%	    {\includegraphics[scale=0.4,page=5]{construction_RD}}
%    \end{minipage} 
    \caption{(a) The graph $C^-_\nu$ obtained from $C_\nu$ in Fig.~\ref{fi:4blocktree}. (b) A transversal pair of bipolar orientations for $C_\nu^-$. (c) The corresponding \rdr. (d) Retracting the rectangles and adding the \se.}
\end{figure}

Consider a 4-connected component $C_\nu$ with outer face $\{u,v,w\}$ in clockwise order.   
We want to find an RVR of $C_\nu$.  To be able to merge later, it also must satisfy one property:  We pick a so-called \emph{\se} $e$ on the outer face beforehand, and we want an RVR such that all edges incident to an endpoint of $e$ are drawn horizontally.

%\todo{TB: note: \se now defined even for $K_4$}
The \se is defined as follows.
For any edge $e$ in triangle $\{u,v,w\}$,
let $x_e$ be the third vertex of the face of $G^+$ that is incident to $e$
and inside the triangle.  Vertex $x_e$ may be original or a dummy-vertex.  
We claim that not all three of $x_{uv},x_{vw}$ and $x_{wu}$ can be 
dummy-vertices.  Assume for contradiction that they are.  If they are all distinct, 
then they would form a \Tc.  If exactly two of them coincide, (say) $x_{uv}=x_{vw}$
is distinct from $x_{uw}$, then in $G$ the crossing at $x_{uw}$ would
form a \Bc with the edge $(u,w)$ involved in the crossing $x_{uv}=x_{vw}$.
Not all three can coincide, else $x_{uv}=x_{vw}=x_{wu}$ would be the only vertex
in the graph rooted at $\nu$ and have degree 3 and not be a dummy-vertex.
%\todo{TB: Rewritten so that this argument works even if $C_\nu$ is $K_4$.}
So there exists at least one edge $e$ of triangle $\{u,v,w\}$ where $x_e$ is an original vertex.  
Define the {\se} of $C_\nu$ to be this edge $e$.  For future reference, we note:

\begin{observation}
\label{obs:special_edge}
Let $e$ be the \se of $C_\nu$, and let $f_e$ be the face of $G^+$ incident to $e$ that is inside the triangle formed by the outer face of $C_\nu$.  Then the vertex of $f_e$ that is not an end of $e$ is an original vertex.
\end{observation}

If more than one edge of triangle $\{u,v,w\}$ satisfies the condition of Observation~\ref{obs:special_edge}, then we chose the \se arbitrarily with one exception: If $\nu$ has a parent $\pi$, and the surround-edge of $C_\pi$ is one of the edges on $\{u,v,w\}$, then we use the same edge also as \se of $C_\nu$.  
Note that whether an edge can be used as \se depends on its incident face in $G^+$ (not $C_\pi$), and so the \se for $C_\pi$, if it exists in $C_\nu$, also qualifies as \se for $C_\nu$.  
For example, for graph $C_\phi$ in Fig.~\ref{fi:4blocktree}, edge $(a,b)$ must be chosen as \se, because both $(b,c)$ and $(a,c)$ have a dummy-vertex as third vertex on the face.  Since edge $(a,b)$ also belongs to $C_\omega$, it becomes the \se of $C_\omega$ as well.

\subsubsection{Drawing a 4-connected component}
%\todo{TB: I made specific what we want of the \se.  This (unfortunately) involved redrawing pictures since we sometimes had it horizontal and sometimes vertical.  Please check carefully that I didn't introduce mistakes here.} 
Now we explain how to find an \RVD of $C_\nu$ such that the \se $e$ and all edges incident to its endpoints are drawn horizontally.  If $C_\nu$ is $K_4$, then such an \RVD
is easily obtained (see also Fig.~\ref{fi:foursc-1}), so we
assume that $C_\nu$ has at least 5 vertices.  This implies that any
inner vertex of $C_\nu$ has degree at least 4, else its neighbors
would form a separating triangle.
Remove the \se from $C_\nu$ to obtain graph $C^-_\nu$.
Since $C_\nu$ is 4-connected, $C_\nu^-$ is internally 4-connected.  We
now obtain an \RVD $\gamma_\nu$ of $C_\nu$ using a \rdr as follows:

%\todo{TB: Changed the claim of the lemma to talk specifically about $C_\nu$, so that we can have the claim about the \se in there.}
\begin{lemma}
%Any internally 4-connected plane graph 
$C'_\nu$ has an \RVD such that any inner vertex of degree 4 satisfies the \foursc and all edges incident to an endpoint of the \se of $C_\nu$ are drawn horizontally.
\end{lemma}
\begin{proof}
Since the graph is internally 4-connected, it admits a transversal pair of bipolar orientations, and with it, a \rdr $\mathcal R$ where blue/red edges correspond to shared vertical/horizontal sides \cite{DBLP:journals/dm/Fusy09,DBLP:journals/tcs/KantH97}.
We choose $v_W$ and $v_E$ to be the ends of the \se, and color the outer edges red.
Fig.~\ref{fi:transversal} shows the transversal pair of bipolar orientations of the planarization of the graph in Fig.~\ref{fi:gnuminus}, and Fig.~\ref{fi:rd} gives the corresponding \rdr.   

Now retract the rectangles a bit as follows.  Let $\varepsilon>0$ be so small that all rectangles have width and height at least $\varepsilon$ and the length of any shared side is at least $\varepsilon$.  Replace each rectangle $[x,x']\times [y,y']$ by a slightly smaller rectangle $[x+\frac{\varepsilon}{3},x'-\frac{\varepsilon}{3}]\times [y+\frac{\varepsilon}{3},y'-\frac{\varepsilon}{3}]$.  By choice of $\varepsilon$ each rectangle still has positive width and height. For each edge of $G$, the shared side of length at least $\varepsilon$ is replaced by a region between its two endpoints, not containing other vertices, that has dimension $2\frac{\varepsilon}{3} \times \frac{\varepsilon}{3}$ or more. Hence for each edge we obtain a line of sight between the rectangles of its endpoints, and hence an \RVD with thick lines of sights.  Note that the line of sight is horizontal if and only if the shared rectangle-sides were vertical, so the horizontal/vertical lines of sights correspond to the red/blue edges of the transversal pair of bipolar orientations.  See also Fig.~\ref{fi:rnu}.

It remains to prove the claim on the \foursc.  Let $z$ be an inner vertex of degree 4, and let $R_z$ be the rectangle of $z$ in the rectangular dual.  By properties of the transversal pair of bipolar orientations, the clockwise order of edges around $z$ consists of incoming red edges, incoming blue edges, outgoing red edges, and outgoing blue edges.  None of these sets is empty for an inner vertex, and by $\deg(z)=4$ hence there exists exactly one edge of each kind.  In consequence, in the \rdr there is exactly one neighbor for each side of $R_z$, which shows the \foursc.
%If the lines of sight of $e_E$ and $e_W$ had some $y$-coordinate in common, then we could shrink the height of $R_z$ to $0$, using only this $y$-coordinate, and still represent $e_E$ and $e_W$ (albeit with zero-thickness lines of sight).  But observe that this is always the case.  For since only one edge of $z$ is incoming red, the entire left side of $R_z$ was shared with the other end of $e_W$, and hence the entire height of $R_z$ is available for the line of sight of $e_E$.  Likewise the entire height of $R_z$ is available for the line of sight of $e_W$, and so these two line of sights have a $y$-coordinate in common.  So we can shrink the height of $R_z$ to 0 and likewise the width of $R_z$ to 0, and $R_z$ becomes a point.
\end{proof}

This gives an \RVD of $C_\nu^-$, to which we need to add the \se.  
Say the outer face of $C_\nu^-$ was $u,v,w,x$, with \se $(u,w)$ and $x$ the 
vertex incident to the inner face at the \se.  
The construction draws all edges at $u$ and $w$ horizontally.   After rotation, assume that $u$ is left of $w$.
We can now extend both $u$ and $w$ downwards beyond the drawing, and insert here a horizontal segment for the \se.  This gives an \RVD $\mathcal R_\nu$ of $C_\nu$.  See also Fig.~\ref{fi:rnu}.  

\begin{observation}
\label{obs:4sidesBase}
In the \RVD of $C_\nu$, the \foursc holds for every dummy-vertex that is an inner vertex of $C_\nu$ and has degree 4 in $C_\nu$.
\end{observation}
\begin{proof}
The claim holds vacuously if $C_\nu$ is $K_4$, so assume it is not.
Let $z$ be a dummy-vertex that is an inner vertex of $C_\nu$.  We know that $\deg(z)=4$ in $G^+$.  This implies that $\deg(z)\leq 4$ in $C_\nu$, and as argued above $\deg(z)\geq 4$ since $C_\nu\neq K_4$.  Deleting the special edge does not change the degree of $z$, therefore $\deg(z)=4$ even in $C_\nu'$.
So if $z$ is an inner vertex of $C_\nu^-$, then the construction for $C_\nu^-$ ensures the \foursc.  The only vertex that is an inner vertex of $C_\nu$, but not of $C_\nu^-$, is the vertex $x$ incident to the inner face at the surround-edge.  However, by Observation~\ref{obs:special_edge} vertex $x$ is not a dummy-vertex, so $z\neq x$. 
\end{proof}

\subsubsection{Zig-zag-slides}\label{ssse:zigzag}

\begin{figure}[ht]
\hspace{\fill}
\includegraphics[width=0.4\linewidth,page=1]{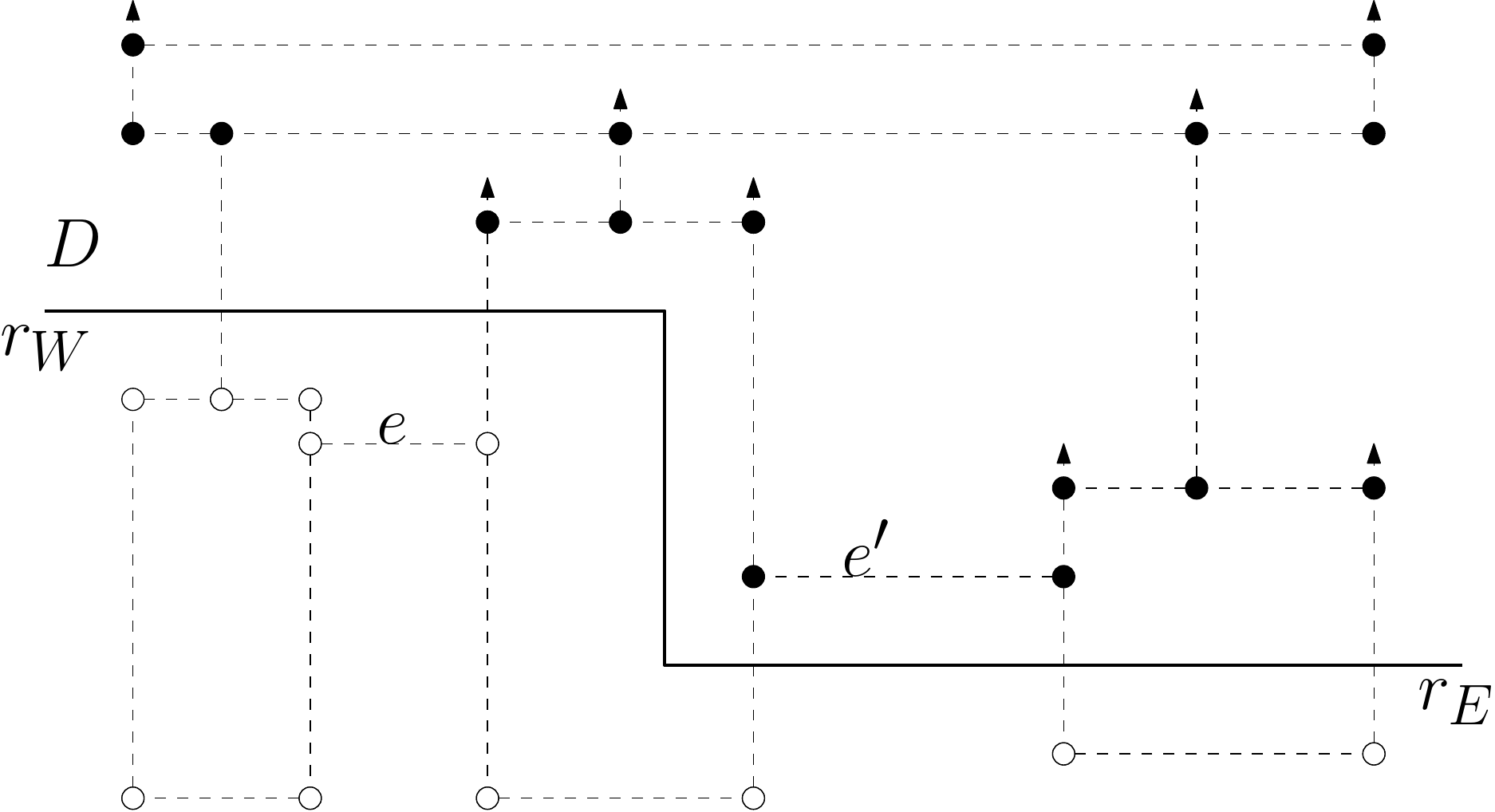}
\hspace{\fill}
\includegraphics[width=0.4\linewidth,page=2]{figures/slide.pdf}
\hspace{\fill}
\caption{A \zz (adapted from \cite{BLPS13}).  Black points move upward while white points remain stationary.  With this amount of sliding, edges $e$ and $e'$ become aligned.}
\label{fi:slide}
\end{figure}

Both for the merging step and for un-doing the planarization, we need a method of modifying the drawing such that some items are moved while others are stationary.  This can be achieved with what was called \zigzag in \cite{BLPS13} (we call it a {\em \zz} for short).   We briefly review this here.
Assume that $\Gamma$ is an \RVD  and we have a dividing curve $D$ as follows:  $D$ consists of some vertical line segment $s$ that intersects no horizontal element of the drawing (i.e., neither a horizontal edge nor a horizontal boundary of a rectangle).  At the top end of $s$, attach a leftward horizontal ray $r_W$.  At the bottom end of $s$, attach a rightward horizontal ray $r_E$.  No conditions are being put onto $r_W$ and $r_E$.  Define the {\em region above $D$} be all points that are in the $x$-range of $r_W$ and strictly above $r_W$, and all points in the $x$-range of $r_E$ and on or above $r_E$.

Now slide all points in the region above $D$ upwards by some amount $\delta>0$.  This maintains a visibility representation, after lengthening vertical edges and rectangle borders suitable, because any horizontal segment either stays stationary or moves in its entirety (recall that no horizontal segments cross $s$).  In particular, for any point $p$ above $r_E$ and any point $q$ above $r_W$ but with a larger $y$-coordinate, a shift by $y(q)-y(p)$ achieves that these points become horizontally aligned.   See Fig.~\ref{fi:slide}.    We will use such {\zz}s for this shape of $D$ as well as for the three other shapes achieved by flipping the picture horizontally and/or rotating $90^\circ$.

\subsubsection{Merging components}\label{ssse:merging}

Let $\mu$ be a child of $\nu$ in $\mathcal T$, and recall that we already obtained drawings $\gamma_\nu$ of $C_\nu$ and (recursively) $\gamma_\mu$ of the graph $G_\mu$ consisting of $C_\mu$ and the components at its descendants.  The common vertices $\{u,y,z\}$ of $G_\mu$ and $C_\nu$ form the outer face of $G_\mu$.  See Fig.~\ref{fi:4blocktree}.  
%The {\em frame} of $C_\mu$, i.e., the outer face of $C_\mu^-$,  
The outer face of $C_\mu^-$,  
consists of $\{u,y,z\}$ as well as a fourth vertex, say $x'$; 
%called the \emph{frame-vertex}; 
after possible renaming of $\{u,y,z\}$ we may assume that the order around the outer face of $C_\mu^-$ is $u,y,x',z$.  In particular, drawing $\gamma_\mu$ contains (up to rotation) node $u$ on the top, node $y$ on the right, the \se $(y,z)$ and node $x'$ at the bottom, and node $z$ on the left.    Let $\gamma_\mu'$ be the drawing obtained from $\gamma_\mu$ by deleting $u,y,z$.  It suffices to merge $\gamma_{\mu}'$, since $u,y,z$ and the edges between them are represented in $\gamma_\nu$ already.
%We modify $\gamma_\mu$ a bit by deleting the edges $(u,y)$, $(y,z)$ and $(z,u)$ (these edges are represented in the parent $\gamma_\nu$); in consequence now all edges incident to $u$ are vertical while all those incident to $y$ and $z$ are horizontal.  It follows that we can ``pull out'' $u,y,z$ arbitrarily far away towards the top/right/left from the rest of $\gamma_\mu$, which we denote by $\gamma_\mu'$. 
See Fig.~\ref{fi:frame}.

Triangle $\{u,y,z\}$ is an inner face of $C_\nu$.
The challenge is now to find a region inside where this face is drawn in $\gamma_\nu$ into which $\gamma_\mu'$  will ``fit''.   Put differently, we want to find a region within $\gamma_\nu$ inside the face $\{u,y,z\}$ that is (after possible rotation) below $u$, to the left of $y$ and to the right of $z$.  We call such a region, a \emph{feasible region} for $\gamma_\mu$ in $\gamma_\nu$.  As we have no control over how triangle $\{u,y,z\}$ is drawn in $\gamma_\nu$, the existence of a feasible region is non-trivial, and may in fact require modifying $\gamma_\nu$ slightly.  We start with a simple case:

\begin{lemma}\label{le:merging-1}
Assume that the \se of $C_\nu$ also belongs to $C_\mu$.  Then 
there exists a feasible region for $\gamma_\mu$ in $\gamma_\nu$.
\end{lemma}
\begin{proof}
If the \se $e$ of $C_\nu$ belongs to $C_\mu$, then by our choice of
\se we also used $e$ as the \se for $C_\nu$.  That is, 
$(y,z)$ is the \se for both $C_\mu$ and $C_\nu$.  In
$\gamma_\nu$, $(y,z)$ was drawn (after possible rotation)
bottommost, with $u$ above it.
A feasible region is now found in the small strip above 
the drawing of edge $(y,z)$.  See Fig.~\ref{fi:merging-1}.
\end{proof}

\begin{figure}[t]
    \centering 
    \begin{minipage}[b]{.2\textwidth}
    	\centering
    	\subfloat[\label{fi:frame}{}]
	    {\includegraphics[scale=0.35,page=1]{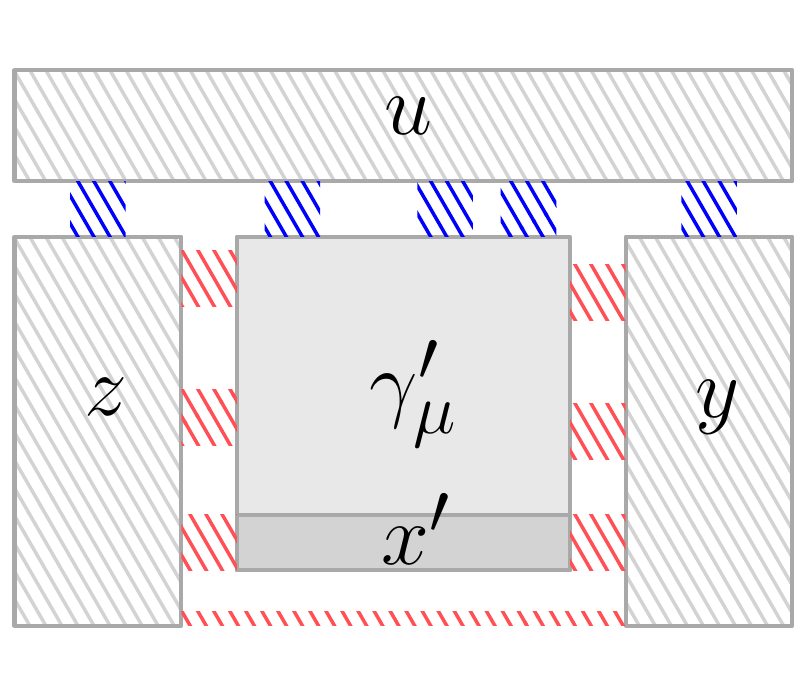}}
    \end{minipage}  
    \hfil
    \begin{minipage}[b]{.2\textwidth}
    	\centering
    	\subfloat[\label{fi:merging-1}{}]
	    {\includegraphics[scale=0.35,page=6]{merging}}
    \end{minipage}  
    \hfil
    \begin{minipage}[b]{.18\textwidth}
    	\centering
    	\subfloat[\label{fi:merging-a}{}]
	    {\includegraphics[scale=0.35,page=2]{merging}}
    \end{minipage}  
    \hfil
    \begin{minipage}[b]{.18\textwidth}
    	\centering
    	\subfloat[\label{fi:merging-b}{}]
	    {\includegraphics[scale=0.35,page=3]{merging}}
    \end{minipage}  
    \hfil
    \begin{minipage}[b]{.18\textwidth}
    	\centering
    	\subfloat[\label{fi:merging-bz}{}]
	    {\includegraphics[scale=0.35,page=4]{merging}}
    \end{minipage}  
    \caption{(a) Illustration of $\gamma_\mu$ and $\gamma_\mu'$.
%After deleting the edges on the outer face of $\gamma_\mu$, we can ``pull out'' $u,y,z$ arbitrarily far away towards the top/right/left from the rest of the drawing. 
(b) Illustration for the proof of Lemma~\ref{le:merging-1}.  
(c--e) Illustration for the proof of Lemma~\ref{le:feasible-region}. 
%Striped regions indicate that the copies of $u,y,z$ in $\gamma_\mu$ can be ``pulled out'' so to be merged with those in $\gamma_\nu$. 
(c) $u$ has two vertical edges. (d) $y$ has two vertical edges, $u$ has a higher top than $z$. (e) Slide so that $u$ has a higher top than $z$. }
\end{figure}

We now show how to find a feasible region for $\gamma_\mu$ in $\gamma_\nu$. 
\begin{lemma}\label{le:feasible-region}
A feasible region for $\gamma_\mu$  in $\gamma_\nu$ always exists,
possibly after modifying $\gamma_\nu$ without changing angles or incidences. 
\end{lemma}
\begin{proof}
We already argued this for the case where $\{u,y,z\}$ includes the surround-edge $e$ of $C_\nu$, so assume that this is not the case.  Then face $\{u,y,z\}$ of $C_\nu$ is also a face of $C'_\nu=C_\nu-e$. 
The drawing of $C'_\nu$ was obtained via the transversal pair of bipolar orientations of $C'_\nu$, with red/blue edges corresponding to horizontal/vertical edges.  It is well-known that any inner face in a transversal pair of bipolar orientations has at least one red and at least one blue edge.    So $\{u,y,z\}$
 is drawn with at least one horizontal and at least one vertical edge, say two edges are vertical and one is horizontal (the other case is similar). Now consider which vertex of $\{u,y,z\}$ is the one where both incident edges of triangle $\{u,y,z\}$ are vertical.

If vertex $u$ has two incident vertical edge $(u,y)$ and $(u,z)$, then $(y,z)$ is drawn horizontally.  We can then find a feasible region within the thick line of sight  of edge $(y,z)$.  See Fig~\ref{fi:merging-a}.

Now assume that $y$ has two incident vertical edge $(y,u)$ and $(y,z)$ (the case of vertex $z$ is symmetric).  For ease of description, assume that $y$ is above $u$ and $z$, and $u$ is to the left of $z$.  If the top side of $u$ is strictly above the top side of $z$, then it is easy to find a feasible region within the thick line of sight of edge $(y,z)$ into which we can merge $\gamma_\nu'$ after rotating it; see Fig.~\ref{fi:merging-b}.  If the top side of $u$ is below (or at the same height) of the top side of $z$, then we modify the drawing to create a suitable region using a \zz. The dividing curve $D$ is defined as follows. Insert a vertical edge just left of $z$, ranging from just above $(u,z)$ (hence below the top of $u$ since we have thick lines of visibility) to just above the top of $z$.  Expand it with a horizontal ray leftward from the lower end and a horizontal ray rightward from the upper end.   By sliding upward sufficiently far, we hence achieve that the top of $u$ is above the one of $z$, which allows us to find a feasible region as previously. See Fig~\ref{fi:merging-bz}.
\end{proof}

By Lemma~\ref{le:feasible-region}, for every child $\mu$ of $\nu$ we can find a feasible region into which we merge (after suitable scaling) the drawing $\gamma_\mu-\{u,y,z\}$. This gives the \RVD of $C_\nu\cup G_\mu$. It remains to prove the \foursc for the inner vertices of $C_\nu \cup G_\mu$.

\begin{observation}
\label{obs:4sidesStep1}
In the \RVD of $C_\nu\cup G_\mu$, the \foursc holds for every dummy-vertex that is an inner vertex of $C_\nu\cup G_\mu$.
\end{observation}
\begin{proof}
Let $v$ be an inner dummy-vertex of $C_\nu \cup G_\mu$.  If $v$ is an inner vertex of either $C_\nu$ or $G_\mu$, then the statement follows by Observation~\ref{obs:4sidesBase} (or induction on the height of $\mathcal T$). Else, $v$ must be on the outer face of $C_\mu$, and hence be part of a separating triangle $\{u,y,v\}$ which is an inner face of $C_\nu$. 

Vertex $v$ must have neighbors both inside and outside triangle $\{u,y,z\}$ by 3-connectivity of the triangulated plane graph $G^+$.  Since dummy-vertices have degree 4, it therefore has exactly one neighbor, say $x$, in $C_\mu-\{u,y,v\}$ and exactly one neighbor, say $z$, in $C_\nu-\{u,y,v\}$.  Further, these neighbors appear in clockwise order $z,u,x,y$ around $v$ and form a kite.  See Fig.~\ref{fi:dummy-sep-2}. It follows that $\{u,y,z\}$ is also a separating triangle (or the outer face of $G^+$) and that $\{v,z,u\}$ and $\{v,y,z\}$ are faces due to the kite.  Therefore $C_\nu$ is a $K_4$, with $\mu$ the only child of $\nu$.  Similarly, $\{u,x,y\}$ is either a separating triangle or an inner face of $G^+$, and $C_\nu$ is also a $K_4$ with at most one child (call it $\omega$ if it exists).

When choosing the \se of $C_\nu$, we cannot use $(u,z)$ or $(y,z)$, because the inner faces at them contain the dummy-vertex $v$. Therefore the \se of $C_\nu$ is $(y,u)$, and by our choice of \se therefore $(y,u)$ also becomes the surround-edge of $C_\mu$ and (if it exists) $C_\omega$.
So Lemma~\ref{le:merging-1} is used to find the feasible region for merging 
%$G_\omega$ into $C_\mu$ and 
$G_\mu$ into $C_\nu$.
Also, the drawing of $C_\nu$ is fixed since it is $K_4$.
Inspecting the construction (see also Fig.~\ref{fi:foursc-2}) shows that during this merge $v$ obtains an incident edge on each side as desired.
\end{proof}

After repeating the process for the entire 4-block tree we obtain the desired \RVD $\Gamma^+$ of $G^+$.  After every merging step inner dummy-vertices of degree 4 satisfy the \foursc.   Once we have merged all subgraphs, all inner dummy-vertices have degree 4, and we get the following result:
%Since the outer face of $G^+$ cannot contain a dummy-vertex, we get the following result.

\begin{observation}
\label{obs:4sidesStep}
In the \RVD of $G^+$, the \foursc holds for every dummy-vertex that is an inner vertex of $G^+$.
\end{observation}

\begin{figure}[t]
    \centering 
%    \begin{minipage}[b]{.24\textwidth}
%    	\centering
%    	\subfloat[\label{fi:dummy-sep-1}{}]
%	    {\includegraphics[scale=0.45,page=2]{dummy-sep}}
%    \end{minipage}  
    \hfil
    \begin{minipage}[b]{.24\textwidth}
    	\centering
    	\subfloat[\label{fi:dummy-sep-2}{}]
	    {\includegraphics[scale=0.4,page=1]{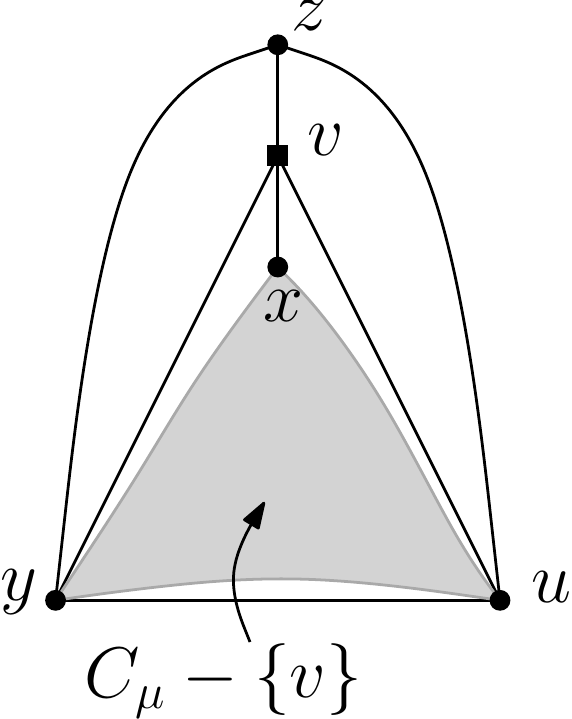}}
    \end{minipage}  
    \hfil
    \begin{minipage}[b]{.24\textwidth}
    	\centering
    	\subfloat[\label{fi:foursc-1}{}]
	    {\includegraphics[scale=0.35,page=1]{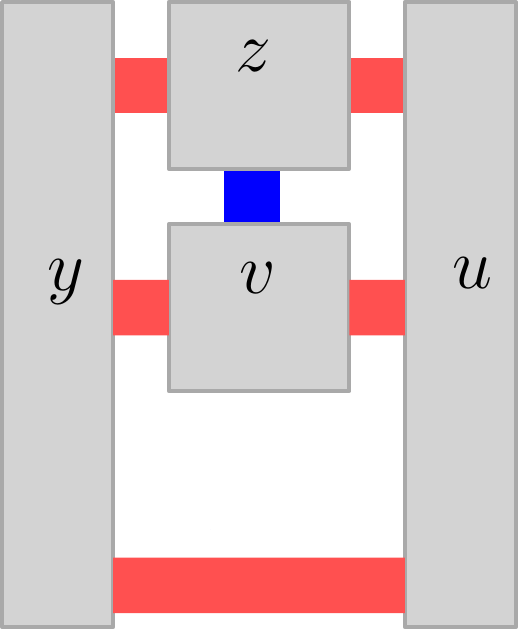}}
    \end{minipage}  
    \hfil
    \centering 
    \begin{minipage}[b]{.24\textwidth}
    	\centering
    	\subfloat[\label{fi:foursc-2}{}]
	    {\includegraphics[scale=0.35,page=2]{dummy-foursc}}
    \end{minipage}  
    \hfil
%    \caption{Illustration for the proof of (a--b) Lemma~\ref{le:dummy-fr} and Lemma~\ref{le:dummy-sep} and for the proof of (b--c) Observation~\ref{obs:4sidesStep}.}
    \caption{Illustration for the proof of Observation~\ref{obs:4sidesStep1}.}
\end{figure}

\subsubsection{Undoing the planarization}\label{ssse:undoing}

Once all merging is completed, the drawing $\Gamma_\rho$ at the root gives
an \RVD of $G^+$.  To turn this into an \RVD of $G$, we must undo the planarization, hence remove the dummy-vertices and replace them with a crossing.  Since dummy-vertices are surrounded by a kite, none of them is on the outer face of $G^+$.  So any dummy-vertex $z$ is an inner vertex of $G^+$ and by Observation~\ref{obs:4sidesStep} its rectangle $R_z$ satisfies the \foursc.  Let $e_W,e_N,e_E$ and $e_S$ be the four edges incident at the west/north/east and south side of $R_z$.  If $e_W$ and $e_E$ have the same $y$-coordinate, and $e_N$ and $e_S$ have the same $x$-coordinate, then we can simply remove $R_z$, extend the edges, and obtain the desired crossing.

So assume that $e_W$ and $e_E$ have different $y$-coordinates, with (say) the $y$-coordinate $y_W$ of $e_W$ larger than the $y$-coordinate $y_E$ of $e_E$.  Construct a dividing curve $D$ by starting with a vertical line segment inside $R_z$, ranging from $y$-coordinate $y_W+\varepsilon$ to $y$-coordinate $y_E-\varepsilon$, where $\varepsilon>0$ is chosen small enough that the segment is inside $R_z$. Attach a leftward horizontal ray at the top and a rightward horizontal ray at the bottom.  Now apply a \zz of length $y_W-y_E$.  This moves $e_E$ upward while keeping $e_W$ stationary, and hence aligns the two edges.  See also Fig.~\ref{fi:slide}, where this is illustrated for $e_W=e$ and $e_E=e'$.
With another \zz (with the dividing curve using a horizontal line segment and vertical rays), we similarly can align $e_N$ and $e_S$, if needed.  After this, remove the dummy-vertex to obtain the crossing.  Repeating this at all dummy-vertices, and finally deleting all added edges, gives the \RVD of $G$, as no crossings are created except where dummy-vertices were removed.
This ends the description of algorithm \algo.

\subsubsection{Graphs that are not 3-connected}\label{ssse:general}

We now explain the (minor) modifications needed to make algorithm
\algo work even for a 1-plane graph $G$ that is not 3-connected.  We
used 3-connectivity only in Section~\ref{ssse:triangulate} where
we planarized $G$ and then triangulated the result to obtain $G^+$.
If $G$ is 3-connected, then $G^+$ is a simple graph.  If $G$ is not
3-connected, then $G^+$ may have multiple edges.  We now explain
how to remove these multiple edges by flipping some of them suitably
in such a way that the resulting graph still has an \RVD that satisfies
the \foursc, 
although some dummy-vertex may have degree greater than four.

So assume $(v,w)$ is a multiple edge in the planar graph $G^+$.
At most one copy of $(v,w)$ may have existed in $G$ as well; we
call this the {\em original copy} and all other copies {\em added
copies}.  For any added copy $e'$, let the two vertices {\em facing}
$e'$ be the two vertices $y,z$ that are not endpoints of $e'$ but are on
a face incident to $e'$.

If for any added copy of $(v,w)$ both vertices facing it are original,
then we can simply {\em flip} the edge, i.e., remove it and connect these
two vertices instead.  This does not create a new multiple edge.
In particular therefore we may assume that for any added copy at
least one facing vertex is a dummy-vertex.

Enumerate the copies of $(v,w)$ as $e_1,\dots,e_s$ in clockwise order
around $v$ in such a way that the 2-cycle $e_1,e_s$ encloses all other
copies.   We say that $e_i$ has a {\em dummy on the left/right side}
if the vertex that faces $e_i$ and is before/after $e_i$ in
the order around $v$ is a dummy-vertex. 
Observe that if $e_i$ has a dummy on the left side, then no $e_j$
for $j<i$ can have a dummy on the right side, else the two corresponding
crossings would form a \Wc.  Also, $e_j$ cannot be the original copy, else 
it and the crossing at $e_i$ would form a \Bc.  So  if $\ell$ is the
maximal index where $e_\ell$ has a dummy on the left side, then
$e_1,\dots,e_{\ell-1}$ all are
added copies and have their crossing on the left side.  (We set $\ell=0$
if there is no such copy.)
Likewise if $r$ is the minimal index such that $e_r$ has a dummy on the
right side, then $e_{r+1},\dots,e_s$ all are
added copies and have their crossing on the right side.
(We set $r=s+1$ if there is no such copy.)
See Fig.~\ref{fig:multiple_edge}.

\begin{figure}[ht]
\hspace*{\fill}
\includegraphics[width=0.3\linewidth,page=1]{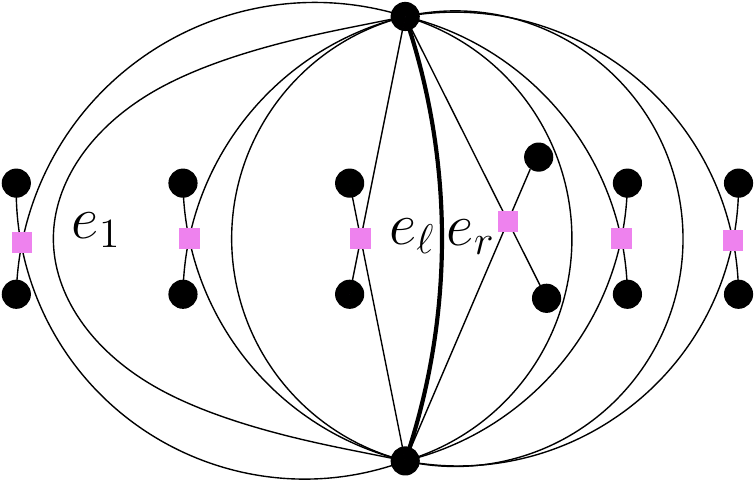}
\hspace*{\fill}
\includegraphics[width=0.3\linewidth,page=2]{figures/multiedge.pdf}
\hspace*{\fill}
\includegraphics[width=0.3\linewidth,page=3]{figures/multiedge.pdf}
\hspace*{\fill}
\caption{A set of multi-edges.  We show $\ell=r$, $\ell=r-1$ and $\ell=r-2$.
The original edge (or the edge declared to be original) is bold.}
\label{fig:multiple_edge}
\end{figure}

Note that $\ell\leq r$ by the above discussion.
We consider possibilities for $r-\ell$:
\begin{itemize}
\item We may have $\ell=r$.  In this case, $e_\ell=e_r$ is the only edge that could be the original copy.  If it was not original, then we add $e_\ell$ to $G$; note that this does not create any \Bc due to our careful choice of $e_\ell$.  We can henceforth assume that the original edge existed.
\item We may have $r=\ell+1$.  In this case, one of $e_\ell$ and $e_r$ may have been the original copy.  If neither of them was original, then as in the previous case we add $e_\ell$ to $G$.
\item We may have $r=\ell+2$.  In this case, edge $e_{\ell+1}$ must have been the original copy, since it has no facing incident dummy-vertex but was not removed by earlier flips.
\item We may not have $r>\ell+2$, else $e_{\ell+1}$ and $e_{\ell+2}$ would both be original, contradicting simplicity of $G$.
\end{itemize}

Let $G'$ be the graph obtained by possibly adding $e_\ell$ to $G$ if we are in one of the first two cases above.  We know that $G'$ has no \Bc,   and it has no \Wc or \Tc either since these cannot be created by adding non-crossed edges.  The very same graph $G^+$ also works as planarized triangulated version of $G'$ (the only thing that has changed is that for every multiple edge there now exists an original copy).  With this, we now have the following properties:
\begin{observation}
Let $(u,v)$ be a multiple edge, and let $e$ be the original copy of $(u,v)$.  Then for every other  copy $e'$ of $(u,v)$, exactly one of the two vertex $y,z$ facing $e'$ is a dummy-vertex.  Moreover, this dummy-vertex forms a separating triangle when combined with $e$, and edge $e'$ is inside this separating triangle.
\end{observation}

Let $G^{+'}$ be the graph obtained from $G^+$ by flipping, for every multiple edge, all added copies.   This cannot create new multiple edges by planarity.  So if we can find an \RVD of $G^{+'}$, then this also gives one of the planarization of $G$ after deleting extra edges.  However, to undo the planarization we must ensure that the \foursc holds, and this is non-trivial since flipping an added copy of a multiple edge results in dummy-vertices whose degree exceeds 4.

We claim that for any such dummy-vertex, the \foursc nevertheless holds automatically when drawing $G^{+'}$ with algorithm \algo.  To see this, let $e'$ be an added copy of an original edge $e$, let $z$ be the dummy-vertex and $y$ be the original vertex that face $e'$.  When flipping edge $e'$, we insert a new edge $(z,y)$.   This edge increases the degree at $z$, but since $y$ is not a dummy-vertex, this is the only way in which degrees at dummy-vertices can be increased.  In particular, $(y,z)$ is inside the separating triangle formed by $z$ and the original edge $e$  and belongs to some child-component $C_\mu$ of this separating triangle. The parent-component $C_\nu$ contains the dummy-vertex $z$ and the four neighbours at the endpoints of the crossing that defined $z$, but it does not contain edge $(y,z)$, or any other edge that may have been added at $z$ due to flipping.  Therefore when creating the \RVD of $C_\nu$, the \foursc is satisfied for dummy-vertex $z$, using (parts of) edges that existed in $G$.  Merging child-components may add more edges at $z$, but no merging step changes edge-directions, and so the four incident edges of $z$ from $G$ remain on the four sides.  Thus if we delete all added edges in the \RVD of $G^{+'}$, then we obtain an \RVD of the planarization of $G$ that satisfies the \foursc at all dummy-vertices, and can hence undo the planarization exactly as before.

\begin{lemma}\label{le:sufficiency}
Let $G$ be a 1-plane graph $G$ and with no \B-/\W-/\Tc as a subgraph. Then algorithm \algo computes an \RVD of $G$.
\end{lemma}

Corollary~\ref{co:nobwt} and Lemma~\ref{le:sufficiency} imply Theorem~\ref{th:characterization}. The next two sections prove Theorem~\ref{th:test}.

\subsection{Area and run-time considerations}

Algorithm \algo does not obviously have linear run-time, since we repeatedly change the entire drawing, especially when applying {\zz}s.  Also, coordinates may get very small when merging, resulting (after re-scaling to be integral) in a very large area.    However, both of these become non-issues if we store {\RVD}s implicitly, using the orthogonal representations of Observation~\ref{obs:convert}.  
Using this approach, no \zz needs to be performed, since these only change edge lengths  (but not angles or adjacencies) and hence give rise to exactly the same orthogonal representation.  Therefore,
all that is required to merge a subgraph at a separating triangle $\{u,y,z\}$ is to find, for each of $w\in\{u,y,z\}$, the appropriate segment along the boundary of the rectangle of $w$ in $C_\mu$, and to merge here the list of adjacencies and angles that $w$ has in the PSLG of $C_\nu$.    This takes constant time for $w$, hence constant time for merging at one separating triangle, and hence linear time overall.  All other parts of the algorithm (such as finding separating triangles, computing 4-connected components, finding a transversal pair of bipolar orientations, and finding the \rdr) take linear time as well.  

The final conversion of the PSLG into an \RVD also takes linear time, and gives linear integral coordinates. Here ``linear'' 
measures the size of the PSLG, which in our case is proportional to the
number of vertices, edges, and crossings.  For 1-planar graphs this
is $O(n)$ since 1-planar graphs have at most $4n-8$ edges and at most $n-2$
crossings \cite{DBLP:journals/combinatorics/CzapH13}.
Therefore
the final \RVD has linear coordinates and the area is $O(n^2)$.  We conclude:

\begin{lemma}\label{le:algo}
Let $G$ be a 1-plane graph $G$ with $n$ vertices and with no \B-/\W-/\Tc as a subgraph. Then there exists an $O(n)$-time algorithm that computes an \RVD $\Gamma$ of $G$ that has $O(n^2)$ area.
\end{lemma}

\subsection{Testing Algorithm}\label{sse:testing}

In order to prove Theorem~\ref{th:test}, we need to show how to test whether a given 1-plane graph $G$ contains any \Bc, any \Wc, or any \Tc. Afterwards, if $G$ contains none of them, we can apply algorithm \algo to produce the desired representation. Hong {\em et al.}~\cite{DBLP:conf/cocoon/HongELP12} show an algorithm to detect every possible \Bc and \Wc in $O(n)$ time, where $n$ is the number of vertices of $G$. Thus, in what follows we describe how to detect a \Tc, assuming that $G$ contains no \Bc and no \Wc. 

Let $G^+$ be the plane triangulated graph obtained from $G$ by applying the planarization and triangulation step described in Section~\ref{sse:sufficient}. Recall that 
%TB: removed this part of the sentence below, because it is not true if G is not 3-connected
%$G^+$ contains no \Bc and no \Wc, and that 
$G^+$ contains a \Tc if and only if $G$ does. 
Suppose $G^+$ contains a \Tc $t$ with outer vertices $a,b,c$ and inner vertices $d,e,f,g,h,i$, as in Fig.~\ref{fi:tgraph}. Since we added a kite around each crossing, vertices $a,b,c$ induce a triangle in $G^+$. 
Thus to detect a \Tc in $G^+$, we first list all triangles of $G^+$.
The triangle is a \Tc if and only the three faces on the inside have dummy-vertices as their third vertex, which can be tested in $O(1)$ time per triangle.  The only difficulty is hence to list all triangles efficiently.  If $G^+$ is simple, then this can be done in $O(n)$ time (see e.g. \cite{Chiba:1985:ASL:3674.3688}).  But if $G$ is not 3-connected, then $G^+$ need not be simple and in particular need not have constant arboricity (the vital ingredient in 
\cite{Chiba:1985:ASL:3674.3688}%
).  Consider again a set of multi-edges as in Figure~\ref{fig:multiple_edge}.  Observe that if any crossing of an added copy was part of a \Tc, then so is either the crossing at $e_\ell$ or the crossing at $e_r$.  Hence to detect a \Tc, we only need to keep the original copy and the added copies $e_\ell$ and $e_r$ before computing all triangles.  We hence need to find all triangles in a planar graph where edges have multiplicity at most 3; this graph has a linear number of edges in any induced subgraph and hence constant arboricity, and so with 
\cite{Chiba:1985:ASL:3674.3688} finding the triangles can be done in linear time.
%Since we have $O(n)$ triangles in  $G^+$, the overall procedure takes $O(n)$ time. 

\begin{lemma}\label{le:test}
Let $G$ be a 1-plane graph with $n$ vertices. There exists an $O(n)$-time algorithm to test whether $G$ has a \Tc, and hence an $O(n)$-time algorithm to test whether $G$ admits an \RVD.
\end{lemma}

Lemma~\ref{le:algo} and Lemma~\ref{le:test} imply Theorem~\ref{th:test}.

\ifvaremb
% ==================================================================
\section{Variable Embedding Setting}\label{se:varemb}
% ==================================================================

\begin{figure}[t]
    \centering
    \begin{minipage}[b]{.3\textwidth}
    	\centering
    	\subfloat[\label{fi:expemb-1}{}]
	    {\includegraphics[scale=0.6,page=1]{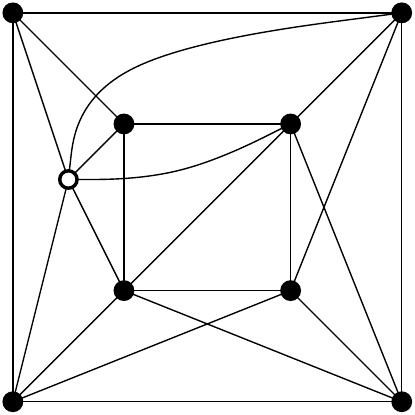}}
    \end{minipage}
    \begin{minipage}[b]{.3\textwidth}
    	\centering
    	\subfloat[\label{fi:expemb-2}{}]
	    {\includegraphics[scale=0.6,page=2]{expemb}}
    \end{minipage}
    \begin{minipage}[b]{.38\textwidth}
    	\centering
    	\subfloat[\label{fi:expemb-3}{}]
	    {\includegraphics[scale=0.4,page=3]{expemb}}
    \end{minipage}
    \caption{
(a) and (b) show two different 1-planar embedding of the same 4-connected 1-planar graph $H$. Attaching $n$ copies of $H$ as in (c) we obtain a 4-connected 1-planar graph with $\Theta(n)$ vertices and $\Omega(2^n)$ different 1-planar embeddings.}\label{fi:expemb}
\end{figure}

Motivated by our characterization, in this section we investigate the variable embedding setting. We remark that, differently from planar graphs, a 1-planar graph can have exponentially many different embeddings, even if 3-connected or 4-connected (see, e.g., Fig.~\ref{fi:expemb}). Thus, an exhaustive analysis may not be feasible. Nevertheless, one may wonder whether being free to change the embedding of a given 1-plane graph $G$ is enough to remove every forbidden configuration. In fact, Alam {\em et al.}~\cite{DBLP:conf/gd/AlamBK13} proved that if $G$ is 3-connected, then it is possible to reroute its edges so as to remove all \Bcs and \Wcs but at most one.  
%Provided that there is no \Tc in $G$, then this would suffice to show the existence of an \RVD $\Gamma$ of $G-e$ (we delete one edge to remove the possible forbidden configuration). Unfortunately, we can prove the existence of an infinite family of 3-connected 1-planar graphs, such that every member of this family contains a linear number of \Tcs in any 1-planar embedding.
We show now that this is not true for \Tc.

\begin{theorem}\label{th:ve-3conn}
Let $k \geq 1$ be an integer value. There exists a 3-connected 1-planar graph $G_k$ with $n=6k$ vertices that contains at least $k-1$ \Tcs in every possible 1-planar embedding.
\end{theorem}

The following stronger corollary can also be proved.

\begin{corollary}\label{co:ve-3conn}
Let $k \geq 1$ be an integer value. There exists a 3-connected 1-planar graph $G_k$ with $n=6k$ vertices that does not admit any 1-planar embedding that can be represented as an \RVD unless we delete at least $k=n/6$ edges.
\end{corollary}

The proofs of Theorem~\ref{th:ve-3conn} and Corollary~\ref{co:ve-3conn} are postponed to Section~\ref{sse:proof-ve3conn}.

\medskip

On the positive side, a 4-connected 1-planar graph $G$ admits an embedding $G'$ with no \Bcs and \Wcs except that the outer face may consist of a crossing with an edge between two of its ends  \cite{DBLP:conf/gd/AlamBK13}.
In any \Tc in a 1-plane graph, the three outer vertices are either the outer face or form a separating triplet, so any embedding of a 4-connected 1-planar graph has at most one \Tc.  This means that in $G'$ there is at most one \Tc, and if there is one then there is no \Bc or \Wc. This leads to the following theorem.

\begin{theorem}\label{th:ve-4conn}
Every 4-connected 1-planar graph $G$ admits a 1-planar embedding that can be represented as an \RVD, except for at most one edge.
\end{theorem}

An {\em optimal 1-planar graph} is one that has the maximum number $4n-8$ of possible edges.
Optimal 1-planar graphs are always 4-connected, and in any 1-planar embedding
all crossings are in kites, with the exception of one crossing on the outer-face
\cite{Suzuki2010}.
% and have a unique embedding (up to choice of the outer face) except for the extended wheel graphs $XW_k$ ($k \geq 6$), which have two different 1-planar embeddings if $k \geq 8$, and six different 1-planar embeddings if $k=6$~\cite{Schumacher,Suzuki2010}. Moreover, every optimal 1-plane graph $G$ can be obtained from a 3-connected plane quadrangulation $Q(G)$ (i.e., a 3-connected plane graph whose faces are all 4-cycles) by adding two (crossed) edges to each face of $Q(G)$~\cite{Suzuki2010}. 
It follows that $G$ does not contain any \Tc, and with the exception of the crossing at the outer face it contains no \Bc or \Wc.

\begin{corollary}\label{co:ve-optimal}
Every optimal 1-plane graph $G$ admits an \RVD, except for one edge.
\end{corollary}

\subsection{Proof of Theorem~\ref{th:ve-3conn} and Corollary~\ref{co:ve-3conn}}\label{sse:proof-ve3conn}

\begin{figure}[t]
    \centering
    \begin{minipage}[b]{.18\textwidth}
    	\centering
    	\subfloat[\label{fi:k5}{$K_5$}]
	    {\includegraphics[scale=0.6,page=7]{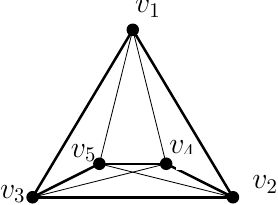}}
    \end{minipage}
    \hfill
    \begin{minipage}[b]{.18\textwidth}
    	\centering
    	\subfloat[\label{fi:k6-1}{$K_6$}]
	    {\includegraphics[scale=0.6,page=8]{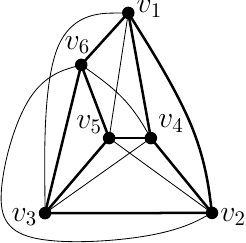}}
    \end{minipage}
    \hfill
    \begin{minipage}[b]{.18\textwidth}
    	\centering
    	\subfloat[\label{fi:k6-2}{$K_6$}]
	    {\includegraphics[scale=0.6,page=9]{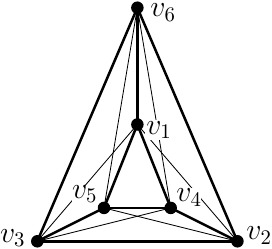}}
    \end{minipage}
    \hfill
    \begin{minipage}[b]{.3\textwidth}
    	\centering
    	\subfloat[\label{fi:gk}{$G_k$}]
	    {\includegraphics[scale=0.4,page=2]{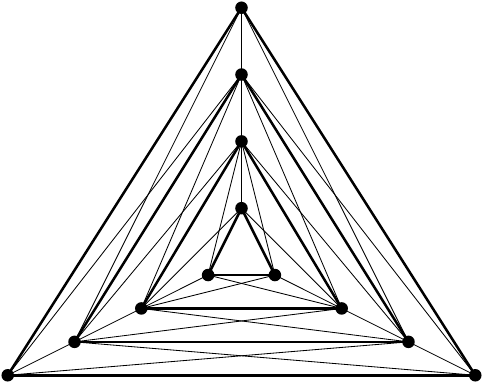}}
    \end{minipage}
    \caption{(a) $K_5$ in its unique (up to renaming and outer face) 1-planar embedding.
(b-c)  $K_6$ in its unique (up to renaming) 1-planar embedding and two possible choices of the outer face. (C) Multiple copies of $K_6$ combined.}
\end{figure}

The idea is to show that $K_6$ has a ``unique'' embedding, up to renaming and choice of the outer face.  To argue this, define
the \emph{planar skeleton} $S(G)$ of a 1-plane graph $G$ to be the plane graph obtained from $G$ by removing all edges that cross some other edge. 

\begin{observation}
For any 1-planar embedding of $K_6$, the planar skeleton is the 3-prism and the outer face forms a \Bc or a \Tc.
\end{observation}
\begin{proof}
Recall that the 3-prism is a 3-connected planar graph whose unique planar embedding (up to renaming and outer face) is shown in Fig.~\ref{fi:k6-1} (the bold edges only).

The complete graph $K_5$ has a unique 1-planar embedding (up to renaming and outer face)~\cite{DBLP:journals/dm/Korzhik08a}, which is shown in Fig.~\ref{fi:k5}. Observe that every face of such embedding is a triangle, either incident to two vertices and a crossing point or to three vertices. In order to realize a 1-planar embedding of $K_6$, a sixth vertex can only be added to one of the four faces of $K_5$ that are incident to three vertices. Two of these four cases are illustrated in Figs.~\ref{fi:k6-1} and~\ref{fi:k6-2}; the other two cases are symmetric. In every one of these two cases the statement holds.
\end{proof}

\begin{corollary}
$K_6$ has no 1-planar embedding that can be represented an an \RVD.
\end{corollary}
\begin{proof}
Consider an arbitrary 1-planar embedding of $K_6$.
Since the planar skeleton of $K_6$ is a 3-prism, every planar edge is incident
to a face of degree 4 in the 3-prism,  and all such faces contain a crossing in $K_6$.  If the outer face contains a crossing then this gives a \Bc
(Fig.~\ref{fi:k6-1}), and if it consists of vertices only then it gives a \Tc  (Fig.~\ref{fi:k6-2}). Either way it does not admit an \RVD by Lemma~\ref{co:nobwt}.
\end{proof}

Now for any $k$ we can create a 3-connected 1-planar graph $G_k$ with $6k$ vertices that consists of $k$ vertex-disjoint copies of $K_6$. Fig.~\ref{fi:gk} shows one possible such construction for $k=2$.  Also, it is immediate to see that all copies of $K_6$ in $G_k$, except for at most one, must be embedded with no crossing points on their outer face (i.e., as in Fig.~\ref{fi:k6-2}). Thus, $G_k$ contains at least $k-1$ \Tcs. This concludes the proof of Theorem~\ref{th:ve-3conn}.

Moreover, if some spanning subgraph $H$ of $G$ had an \RVD $\Gamma$, then for each of the $k$ copies of $K_6$, at least one edge cannot exist in $\Gamma$, else $\Gamma$ would induce an \RVD of $K_6$.  Hence at least $k=n/6$ edges of $G$ are not represented in $\Gamma$. This proves Corollary~\ref{co:ve-3conn}.

\fi

% ==================================================================
\section{Conclusions and Open Problems}\label{se:conclusions}
% ==================================================================

In this paper, we studied \rvds of non-planar
graphs with a fixed embedding.  We showed that we can test in polynomial time
whether a graph has such a representation.  Of special interest to us were
1-planar graphs; here we can give  a linear-time 
algorithm to test the existence of visibility representations if the embedding
is fixed.  Moreover, in case of a negative answer the algorithm provides a 
witness in form of either a \Bc, \Wc, or \Tc.  
\ifvaremb
We also briefly studied
1-planar graphs without fixed embeddings; we showed that not all
3-connected 1-planar graphs have visibility representations with at most one crossing
per edge,  but all 4-connected 1-planar graphs do after deleting one edge.
\fi

The most pressing open problem is whether we can restrict the drawings less and still test for the existence of visibility representations?  Most importantly, if we fix the rotational scheme and outer face, but {\em not} order in which crossing occur (and perhaps not even which edges cross), is it possible to test whether a visibility representation respecting the rotational scheme and outer face exists?  The NP-hardness proof of Shermer \cite{DBLP:conf/cccg/Shermer96} does not hold if the rotational scheme is fixed, since it uses a reduction from linear-arboricity-2, and fixing the rotational scheme would severely restrict the possible ways of splitting a graph into two linear forests.  The orthogonal representation approach utterly fails if the order of crossings is not fixed, since it requires planarization as a first step.

Secondly, can we characterize for more graph classes exactly when they have a rectangle visibility representation?
Notice that Lemma~\ref{le:outerface} does not use 1-planarity, and hence gives a necessary condition for any subgraph 
of a graph $G$ (with a fixed embedding) to have an \RVD.  A second necessary condition stems from that as soon as edges
may have 2 or more crossings, the edge-parts between the crossings may be non-trivial and have
cycles; clearly the graph formed by these crossings must be bipartite if $G$ has a visibility representation.
Are these two conditions sufficient, and if not, can we find necessary and sufficient conditions, at least
for some restricted graph classes such as 2-planar and fan-planar graphs?

Finally, it is immediate to observe that any \RVD can be transformed into a drawing such that vertices are points, edges are polylines with at most 2 bends, and crossings occur only at right angles. Hence, Lemma~\ref{le:algo} can be used to compute {\em RAC-drawings} (see e.g.~\cite{dl-cargd-12}) with at most 2 bends per edge in linear time and quadratic area.  Since not all 1-plane (straight-line drawable) graphs admit a bendless RAC-drawing~\cite{el-rac1p-DAM13} (and when they do they may require exponential area~\cite{BrandenburgDEKL15}), we ask whether all 1-plane graphs admit a RAC-drawing with at most 2 bends per edge in polynomial area?

\bibliography{paper}
\bibliographystyle{abbrv}

\end{document}